\documentclass[12pt, a4paper]{extarticle}
\usepackage[margin=2cm]{geometry}
\usepackage{authblk}
\usepackage{titling}

\usepackage[english]{babel}
\usepackage{graphicx} % Required for inserting images
\usepackage{wrapfig}
\usepackage{amsmath,amsthm,amssymb,scrextend}
\usepackage{booktabs}            % professional-quality tables
\usepackage{multirow}            % tabular cells spanning multiple rows
\usepackage{amsfonts}            % blackboard math symbols
\usepackage{multicol}
\usepackage[shortlabels]{enumitem}
\usepackage{cutwin}
\usepackage{float}
\usepackage{csquotes}
\usepackage{enumitem}
\usepackage{tikz}
\usepackage{hyperref}

\def \rR {\mathbb{R}}
\def \bf#1 {\textbf{#1 }}
\def \sumt {\sum\limits}
\providecommand{\keywords}[1]
{
  \small	
  \textbf{\textit{Keywords:}} #1
}

\renewenvironment{proof}{\begin{addmargin}[1em]{0em}\begin{newproof}}{\end{newproof}\end{addmargin}\qed}
\newtheorem{thm}{Theorem}
\newtheorem*{thm*}{Theorem}
\newtheorem{lm}{Lemma}
\newtheorem{cor}{Corollary}
\newtheorem*{lm*}{Lemma}
\newtheorem{defin}{Definition}
\newtheorem*{defin*}{Definition}

\def \sumt {\sum\limits}

\def \cN {\mathcal{N}}

\def \l {\lambda}

\def \rR {\mathbb{R}}
\newcommand{\mint}{\min\limits}

\begin{document}

\title{Capability centrality: the next step from scale-free property}
\author{Mikhail Tuzhilin$^1$%
  \affil{$^1$HSE university, Moscow State University}
}
\date{}
\maketitle

\begin{abstract}
In this article we present a new centrality measure called ksi-centrality. We show that ksi-centrality distinguishes real networks from random ones, similar to degree centrality: the ksi-centrality distribution is right-skewed for real networks and centered for random Erdos-Rényi networks, and has linear pattern with a heavy tail on a log plot. Furthermore, the ksi-centrality distribution is centered for models simulating real networks: Barabasi-Albert, Watts-Strogatz, and Boccaletti-Hwang-Latora. Thus, this centrality distribution is an additional and independent property with respect to scale-freeness. We also introduce a normalized version of ksi-centrality and show that it is related to algebraic connectivity and the Chegeer's value of a network. Moreover, the average value of this normalized centrality is in bijective correspondence with the relative number of edges that a new node connects to others in the Barabasi-Albert preferential attachment model, thus answering the question of how to choose the parameter $m$ to model a given real-world network.

\end{abstract}


%------------------------------------------------------------------

\section{Introduction}

Centrality measures play the key role in the networks analysis starting with the famous experiment of two social psychologists Stanley Milgram and Jeffrey Travers in 1969 year~\cite{MT}. They discovered the six degrees of separation law, which formed the basis of the small-world theory for real networks~\cite{Watts}. In 1999, Albert, Zhong, and Barabasi discovered another property of real networks—the distribution of nodes degrees—that distinguishes real networks from random ones~\cite{BA1}. They found that for real networks the node degree distribution is right-skewed, almost linear on the log-log plot, and follows a power law, while for random networks it is central and follows a binomial law. While the question whether the degree distribution of real networks satisfy the power law is currently quite controversial~\cite{Notrare}, the main difference between real and random networks remains the same: the skewness of this distribution is different for real and random networks. In addition, A. Broido and A. Clauset showed that the Weibull distribution is better suited to describe the distribution of nodes degrees in real networks~\cite{Notrare}.

In this paper, we found that Pearson’s moment coefficients
of skewness can be used as a good identifier to distinguish real networks from artificial ones in terms of distributions. A simple threshold of 1 accomplishes this task. It works for degree distribution and a new centrality distribution we introduced called ksi-centrality. To reach this conclusion, we analyzed 40 networks based on real-world data from various fields: biological, social, internet, transportation, and so on. These networks varied in size, structure, density, and other parameters. It turned out that the ksi-centrality distributions for all these networks are right-skewed, close to linear, with a heavy tail on the logarithmic plot, and are well approximated by the Weibull distribution.

Another important question in real-world network science is how to model the properties of real-world networks. In 1998, Watts and Strogatz proposed a model for simulating the small-world property, but it turned out that this model does not satisfy the scale-free condition~\cite{Watts}. In contrast, in 1999, Barabasi and Albert proposed a model that satisfies the scale-free condition but does not satisfy the small-world property~\cite{BATWO}. Finally, in 2007, S. Boccaletti, D-U. Hwang, and V. Latora. proposed a suitable model that is both small-world and scale-free~\cite{BHL}. In this article we analyzed the ksi distributions for these three models. It turned out that the ksi distributions for all of these models are also centered for most parameter values, as for the random Erdos-Renyi model. Therefore, ksi distribution is not only independent of the shape of degree distribution but also provides us with an additional property of real networks that distinguishes them from these models. In contrast, a centrality similar to ksi-centrality --- average neighbor degree --- depends on the degree distribution and therefore lacks the same properties for distinguishing networks. We present this comparison in Section~\ref{ksiVSand}. In addition, normalized version of ksi centrality has connections with algebraic connectivity, the Chegeer’s value for a network, and determines the $m$ parameter for the Barabasi-Albert model simulating a given real network.

\section{Ksi-centrality and ksi-coefficient}

Consider a network of you and your friends, so-called an ego network. Let's ask a question: what is a measure of a capability to expand the circle of your acquaintances? The word "capability" is a combination of two words: "capacity" and "ability", therefore the meaning of this question is to determine a capacity of your own ability to expand your ego network. The simplest way to do that is to follow the rule: a friend of a friend is a friend of mine --- ask a friend to introduce me to a friend of his whom I don’t know yet. Thus, this capability is the average number of friends of my friends, excluding themselves. We call this capability a ksi-centrality of a node. More formally, consider a connected undirected network $G$ with $n$ nodes. Let $\cN(i)$ be the neighborhood of vertex $i$, and for any two disjoint subsets of nodes $H, K\subset V(G)$ let's denote the number of edges with one end in $H$ and another in $K$ by $E(H, K) = \big|{(v,w): v\in H,\ w\in K}\big|$.

\begin{defin}
    Ksi-centrality $\xi_{i}$ of a node $i$ is the relation of the total number of neighbors of $i$'s neighbors excluding themselves divided by the total number of neighbors of $i$:
    $$  
        \xi_i = \xi(i) = \frac {\Big|E\big(\cN(i), V\setminus \cN(i)\big)\Big|} {\big|\cN(i)\big|} = \frac {\Big|E\big(\cN(i), V\setminus \cN(i)\big)\Big|} {d_i},
    $$ where $d_i$ is a degree of the vertex $i$.
\end{defin}

Since by this definition the vertex $i\in V\setminus \cN(i)$, in the sum we calculate vertex $i$ $d_i$ times. Therefore, more precisely, this capability is equal to $\xi_i-1$. Thus, it differs from ksi by a constant, but this will not be important, since further we will consider distributions of these quantities. It is easy to see that ksi-centrality is bounded by $1$ and $n-d_i$. Let's also define the normalized version of this centrality, called normalized ksi. 

\begin{defin}
    Normalized ksi-centrality $\hat\xi_{i}$ of a node $i$ is defined by following
    $$  
        \hat{\xi}_i = \hat{\xi}(i) = \frac {\Big|E\big(\cN(i), V\setminus \cN(i)\big)\Big|} {\big|\cN(i)\big|\cdot\big|V\setminus \cN(i)\big|} = \frac {\Big|E\big(\cN(i), V\setminus \cN(i)\big)\Big|} {d_i (n-d_i)}.
    $$
\end{defin}

We call the average of the normalized ksi-centrality by the average normalized ksi-coefficient $\hat\Xi$, which can be interpreted as the average capability of a node in a network. We found that the average normalized coefficient ksi $\hat\Xi$ does not depend on the network size $n$, and it is in bijective correspondence with the relative number of edges by which a new node connects to others in the preferential-attachment Barabasi–Albert model~\cite{BATWO}, starting with the star graph (Fig.~\ref{fig:ksiVSrel}, supp. fig.~\ref{fig:ksirel}  left). Thus, this correspondence answers the question: if I have a network, how do I find the parameters in the Barabasi-Albert model to build a network similar to mine? The answer is the following: calculate the average normalized coefficient ksi $\hat\Xi$ and find the corresponded $m$ by using the fitted function (Fig.~\ref{fig:ksiVSrel}, right). These will be the parameters for the desired Barabasi-Albert model.

\begin{figure}[h!]
    \centering
	\includegraphics[width = 0.8\textwidth]{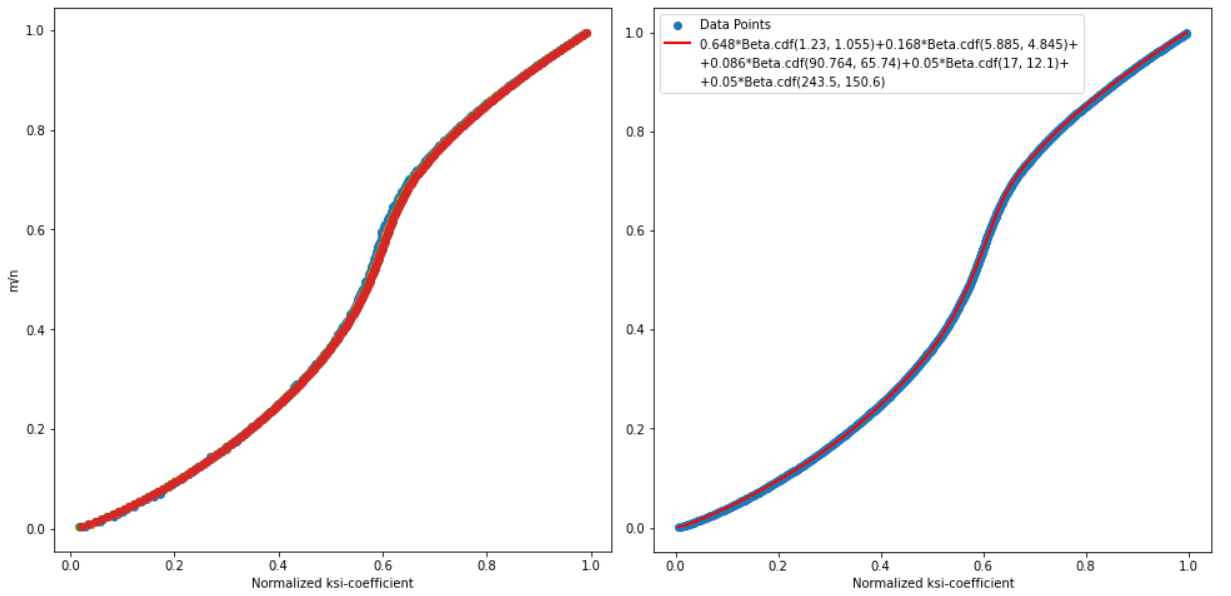}
	\caption{Dependency of normalized ksi-coefficients vs relative number of attachments $m/n$ of for Barabasi-Albert networks $(m,n)$ for fixed $n = 200, 500, 750, 1000$ (left) fitted with the sum of 5 cumulative distribution functions of a beta distributions (right).}
	\label{fig:ksiVSrel}
\end{figure}

We have shown that for a random Erdos-Rényi graph $G(n,p)$ the expected value of average normalized ksi-coefficient tends to $p$ as the number of nodes increases (Thm.~\ref{thm1}). Thus, for random networks we have $p$, equal to the probability of connecting a selected node to another, $p$ --- the internal connectivity of the ego network, i.e. expected value of the average clustering coefficient and also $p$ --- the external connectivity of the ego network, i.e. the average normalized ksi-coefficient $\hat\Xi$ for networks with a sufficiently large number of nodes. 

\begin{thm*}
    For any vertex $i$ of an Erdos-Renyi graph $G(n,p)$ the expected number 
    $$
        \mathbb{E}\big(\hat\xi_i\big) = \mathbb{E}\big( \hat\Xi(G)\big) =  p\Big(1-(1-p)^{n-1}\Big)+\frac {1-p^n} n.
    $$
\end{thm*}

\section{Distributions of ksi-centralities}

For real networks, we found that both normalized ksi-centrality and normalized ksi-coefficient tend to decrease with increasing number of nodes. However, it turned out that for networks with a large number of nodes, both ksi-centrality and normalized ksi-centrality exhibit similar distributions (Supp. fig.~\ref{fig:ksiVSksi}), therefore further we are only calculating ksi-centrality distributions for the most cases.

We have analyzed 40 different real-data networks from different fields: biological, social, internet, transportation and so on. These networks are listed in the supplementary table~\ref{tab:net}. For all these networks, ksi-centrality distributions exhibited properties similar to the node degree distributions: they are right-skewed, looks linear with heavy-tailed on a log plot (Fig.~\ref{fig:lin}), as is the case for the scale-free property, but on a log-log plot (Supp. figs.~\ref{fig:1}--\ref{fig:5}).  A. Broido and A. Clauset showed that the log-normal and Weibull distributions better describe the degree distribution in real networks than the power-law distribution~\cite{Notrare}. On the supplementary figure~\ref{fig:wb} we show, that ksi distributions are also well fitted by the Weibull distribution. 

\begin{figure}[H]\vspace{-10pt}
    \centering
	\includegraphics[width = 0.8\textwidth]{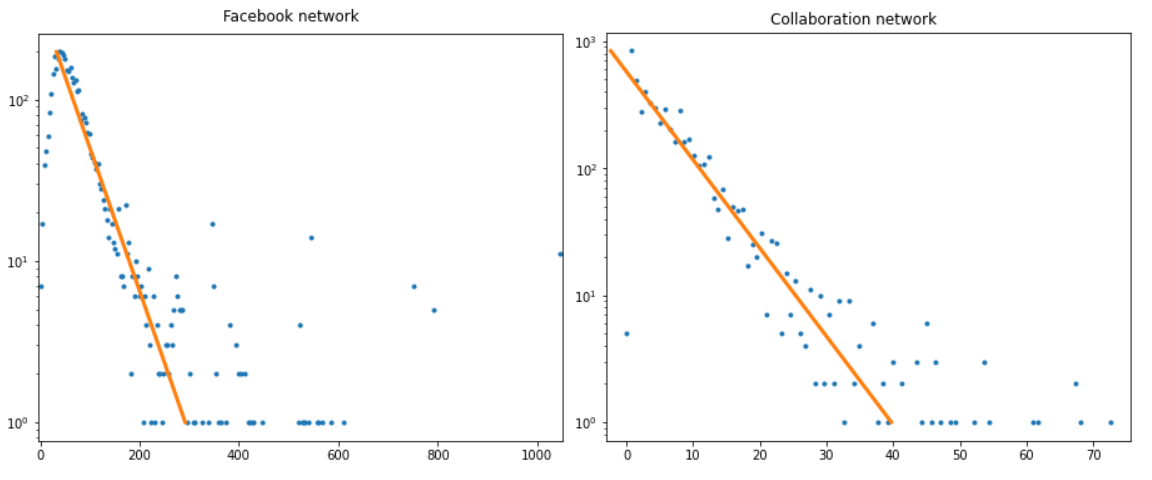}
	\caption{Linear fitting of ksi distribution at logy plot for networks from supplementary table~\ref{tab:net}. The x-axis shows the ksi values, and the y-axis shows the number of vertices.}
	\label{fig:lin}
\end{figure}

For random networks, as with the scale-free property, the ksi distributions were found to be centered, bell-shaped and fitted well by the Weibull distribution (Fig.~\ref{fig:weibart}). We use Pearson's moment coefficient of skewness to calculate the skewness of distributions:

$$
    \gamma_1 (\xi) = \frac {\frac 1 n \sumt_{i\in V(G)} (\xi_i-\Xi)^3} {\Big(\frac1 n \sumt_{i\in V(G)} (\xi_i-\Xi)^2 \Big)^{3/2}},
$$
where $\Xi$ is the average ksi value.

We found that there is a gap in terms of coefficient of skewness between real networks and artificial (Watts-Strogatz, Barabasi-Albert and Boccaletti-Hwang-Latora models). For real networks the coefficient of skewness is greater than 1, and for artificial ones it is less for most parameter values (see supp. fig.~\ref{fig:wb} for real and fig.~\ref{fig:weibart} for artificial). However, if all parameters in the Barabasi-Albert and Watts-Strogatz models are sufficiently small, then the distributions of ksi also turn out to be right-skewed (see supp. figs.~\ref{fig:BA}--\ref{fig:WS2}).

\begin{figure}[h!]
    \centering
	\includegraphics[width = 0.9\textwidth]{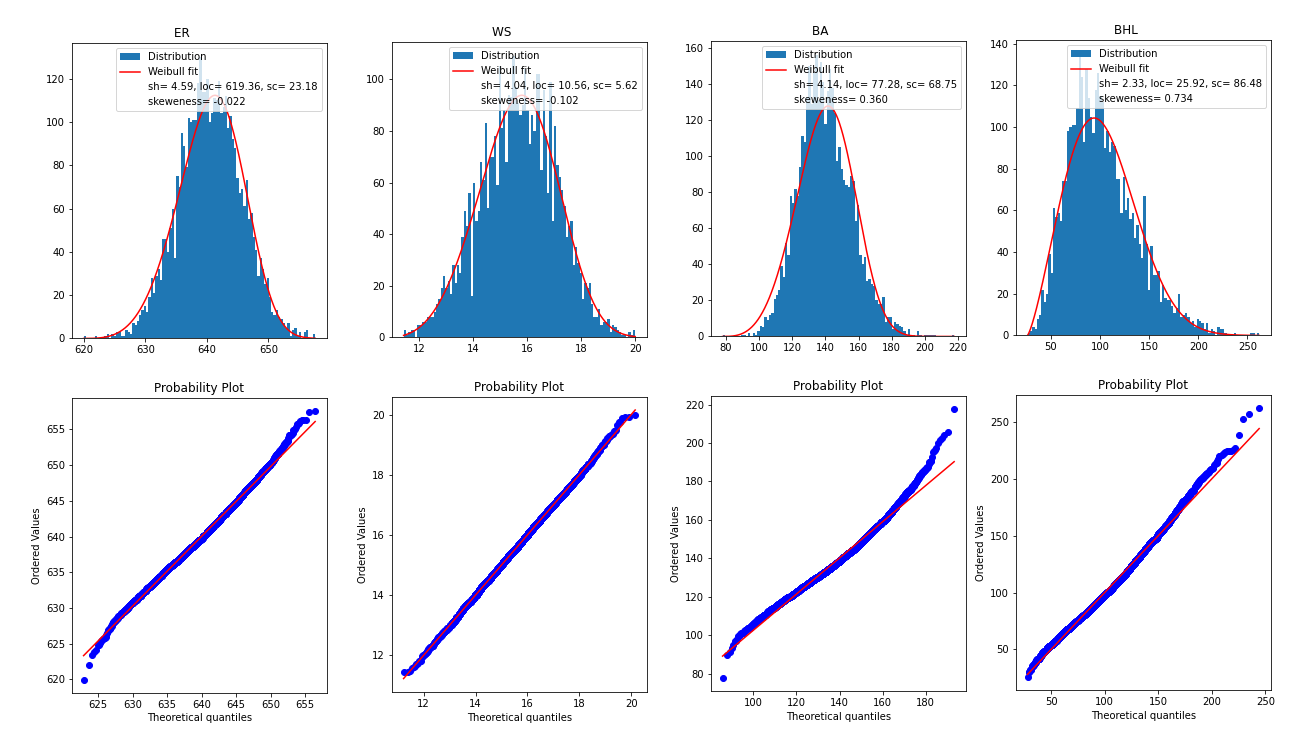}
	\caption{Ksi-distributions fitted with Weibull distribution, their Pearson's moment coefficients of skewness (top) and corresponded QQ-plots (bottom) for networks from models Erdos-Renyi $(n=4000, p=0.2)$, Watts-Strogatz $(n=4000, k=20, p=0.3)$, Barabasi-Albert $(n=4000, m=43)$ and Boccaletti-Hwang-Latora $(n_0=100,  m=20, n=4000)$. The x-axis shows the ksi values, and the y-axis shows the number of vertices for distributions.}
	\label{fig:weibart}
\end{figure}

Therefore, this coefficient can be used as indicator of artificial networks. Also, this threshold works well to differentiate networks in terms of degree distributions (see supp. fig.~\ref{fig:weib_degs}).  

\section{Why ksi but not average neighbor degree?}\label{ksiVSand}

There is a simple relation between the average neighbor degree $k_{nn,i}$, ksi-centrality $\xi_i$, degree $d_i$ and the local clustering coefficient $c_i$. Since the local clustering coefficient is the fraction of the number of edges in $\cN(i)$,
$$
     k_{nn,i} = \xi_i+c_i(d_i-1).
$$
The relationship between the distribution of local clustering coefficients and the distribution of vertex degrees in real-world networks remains unclear. Nevertheless, many of them are small-world networks and thus, have relatively large average clustering coefficients. Therefore, the average neighbor degree depends more on the degree distribution than the distribution of ksi-centrality. However, since the ksi-centrality distribution appears to have properties similar to a scale-free distribution, adding a scale-invariant distribution to it does not significantly alter its shape. We see it on the supplementary figures~\ref{fig:ad1}--\ref{fig:ad5}. Thus, for real networks it exhibits similar properties, but in networks where the degree distribution law is expressed more clearly (for example, in the Barabasi-Albert model), we see this influence clearly (see supp. fig~\ref{fig:adex}). In summary, the average neighbor degree distribution exhibits properties similar to the ksi-distribution for real networks, but does not distinguish real networks from random ones as well as the ksi-centrality distribution does.

\section{Discussion}

In this article we proposed a new measure of centrality called ksi-centrality. This centrality measure can be interpreted as the capability or potential of each node to increase the number of its acquaintances. We have shown that the ksi-centrality distribution has properties similar to those of the node degree distribution: it is well approximated by the Weibull distribution, is right-skewed for real networks, and is centered for random networks. Furthermore, ksi-centrality distributions for real networks exhibit a "linear" behavior with a "heavy tail" on the log plot (the degree distribution has the same property on a log-log plot). We found that the "tail attachment point" on the log plot can be well approximated by the mean logarithm of the distribution.

It turns out that, in contrast to the degree distribution, the distribution of ksi-centrality is centered for well-known models simulating real networks: the Watts-Strogatz, Barabasi-Albert, and Boccaletti-Hwang-Latora models for the most parameter values. Therefore, this distribution is independent of the shape of the vertex degree distribution and classifies these models into the class of artificial networks (as a random network) for most parameter values.

The definition of ksi centrality is very similar to the definition of average neighbor degree, nevertheless, the average neighbor degree depends more heavily on the degree distribution, and therefore the shape of this distribution does not allow one to distinguish real networks from models with a more pronounced degree distribution, such as the Barabasi-Albert model.

We found that ksi-centrality and its normalized version have very similar distributions; however, normalized ksi-centrality appears to approach 0 as the number of nodes increases. Therefore, ksi-centrality will be more convenient for most applications. Nevertheless, normalized ksi-centrality is related to the algebraic connectivity of the network (the second-smallest eigenvalue of the Laplace matrix) and the Chegeer's value (Supp. thms.~\ref{thm:nksi1}, \ref{thm:nksi2}). The average value of normalized ksi-centrality, which can be interpreted as the average capability in a network, is independent of the network size in the Barabasi-Albert model and is in bijective correspondence with the ratio $m/n$ ---the relative number of edges by which a new node connects to others. Therefore, it provides an answer to the question of how to find an appropriate value of $m$ if you want to model your network using the Barabasi-Albert model.

Finally, turn to the interpretation of ksi as a centrality measure --- specifically, what it captures about node importance in the network structure. By definition, ksi-centrality allows one to determine an important node based on the influence of its neighbors, even if those neighbors are not familiar with each other. A node with high ksi-centrality might represent a ruler who has many contacts, and his contacts also have many contacts, or a "man in a gray suit" who has contacts with the most powerful people, and they may be strangers to each other. Supplementary theorem~\ref{thm:star} states that in a star graph, peripheral vertices are more important than the central vertex in terms of ksi-centrality, since their neighbors have more "power" than the neighbors of the central vertex. Therefore, this centrality measure does not satisfy the Freeman star property~\cite{Free}. We call this ksi centrality only because it is a function on the nodes, just like the clustering coefficient. Moreover, like the clustering coefficient, its normalized version is constant on the vertices of the star graph, has the same computational complexity (Supp. cor.~\ref{cor1}) and is equal to $p$ for a random Erdos-Renyi graph with a large number of nodes.

\section{Code availability}
Code is open source and available at \url{https://github.com/Samoylo57/ksi-centrality}.

\section{Acknowledgments}
We thank Ivan Samoylenko for useful comments and development of the code for calculation of ksi-centrality and Aleksandr Levin for the excellent data fitting. The support from the Basic Research Program of HSE University is gratefully acknowledged.

\newpage
\title{Supplementary Material \\
\emph{Capability centrality: the next step from scale-free property}}
\maketitle

In what follows we consider simple undirected graphs $G$ with $A = A(G) = \{a_{ij}\}$ --- adjacency matrix and by $L = L(G) = \{l_{ij}\}$ --- Laplacian matrix.

Since $\frac {\Big|E\big(\cN(i), V\setminus \cN(i)\big)\Big|} {d_i} = \frac {d_i} {d_i} = 1$, when vertices of $\cN(i)\cup\{i\}$ have no adjacent vertices except themselves, we define ${\xi}_i = 1$ in the case when $d_i = 0$. Also note that vertex $i\in V\setminus \cN(i)$, thus the ksi-centrality $\xi_i$ is always greater than or equal to 1. Therefore, we give additional normalized definition --- normalized ksi-centrality $\hat{\xi}$. For this centrality $\frac 1 {n-d_i}\leq\hat{\xi}_i\leq 1$. Since $\frac {\Big|E\big(\cN(i), V\setminus \cN(i)\big)\Big|} {d_i(n-d_i)} = \frac {d_i} {d_i(n-d_i)} = \frac 1 {n-d_i}$, when vertices of $\cN(i)\cup \{i\}$ have no adjacent vertices except themselves, we define $\hat{\xi}_i = \frac 1 n$ for the case, when $d_i = 0$. 

Next we show that for an Erdos-Renyi graph $(n,p)$ the expected value of both normalized ksi-centrality and the average normalized ksi-coefficient equals to almost $p$ for large $n$. Let's first prove

\begin{thm}
    For any vertex $i$ of an Erdos-Renyi graph $G(n,p)$ the expected value
    $$
        \mathbb{E}\Big(\Big|E\big(\cN(i), V\setminus \cN(i)\big)\Big|\Big) = p(n-1)(1+p(1-p)(n-2)).
    $$
\end{thm}
\begin{proof}
    Denote the random variable $e = \big|E(\cN(i), V\setminus\cN(i))\big|$. Note that $P(d_i = k) = C_{n-1}^kp^k(1-p)^{n-1-k}$. Since the maximum number of edges from $\cN(i)$ to $V\setminus \cN(i)\setminus\{i\}$ equal to $k(n-1-k)$, thus $P(e = t+k\,|\,d_i = k) = C_{k(n-1-k)}^t p^t (1-p)^{k(n-1-k)-t}$. Let's denote $f(k) = k(n-1-k)$. The expected value 
    $$
        \mathbb{E}(e) = \sumt_{k = 0}^{n-1}\sumt_{t = 0}^{k(n-1-k)} (t+k)\, P(e = t+k) =  \sumt_{k = 0}^{n-1}\sumt_{t = 0}^{f(k)} (t+k)\, P(e = t+k\,|\,d_i = k)\, P(d_i = k) =
    $$
    $$
        = \sumt_{k = 0}^{n-1}\sumt_{t = 0}^{f(k)} (t+k)\, C_{f(k)}^t p^t (1-p)^{f(k)-t} C_{n-1}^kp^k(1-p)^{n-1-k} = 
    $$
    $$
       = \sumt_{k = 0}^{n-1} C_{n-1}^k p^k (1-p)^{n-k-1}\sumt_{t = 0}^{f(k)} (t+k)\,C_{f(k)}^t(1-p)^{f(k)-t} p^{t} = 
    $$
    $$
        = \sumt_{k = 0}^{n-1} C_{n-1}^k p^k (1-p)^{n-k-1}\Big(k+\sumt_{t = 1}^{f(k)} t\,C_{f(k)}^t(1-p)^{f(k)-t} p^{t}\Big)
    $$
    
    Note that $\sumt_{t = 0}^{f(k)}\,C_{f(k)}^t(1-p)^{f(k)-t} p^{t} = (p+1-p)^{f(k)} = 1$. Also $n(x+y)^{n-1}= \Big((x+y)^n\Big)_x = \Big(\sumt_{t = 0}^nC_{n}^t x^{t}y^{n-t}\Big)_x = \sumt_{t = 1}^n t C_{n}^t x^{t-1}y^{n-t}$, thus
    $$
        \sumt_{k = 0}^{n-1} C_{n-1}^k p^k (1-p)^{n-k-1}\Big(k+\sumt_{t = 1}^{f(k)} t\,C_{f(k)}^t(1-p)^{f(k)-t} p^{t}\Big) = \sumt_{k = 0}^{n-1} C_{n-1}^k p^k (1-p)^{n-1-k}\big(k+p f(k)\big) = 
    $$
    $$
        = p (n-1) + \sumt_{k = 0}^{n-1} C_{n-1}^k p^{k+1} (1-p)^{n-1-k} k (n-1-k) = 
    $$
    $$
        = p (n-1) + p^2(1-p) \sumt_{k = 1}^{n-2} C_{n-1}^k p^{k-1} (1-p)^{n-2-k} k (n-1-k) =p(n-1)+p^2(1-p)(n-1)(n-2),
    $$
    using the same procedure for $(n-1)(n-2)(x+y)^{n-3}= \Big((x+y)^{n-1}\Big)_{xy} = \sumt_{t = 1}^{n-2} t (n-1-t) C_{n-1}^t x^{t-1}y^{n-2-t}$.

\end{proof}

\begin{thm}\label{thm1}
    For any vertex $i$ of an Erdos-Renyi graph $G(n,p)$ the expected number 
    $$
        \mathbb{E}\big(\hat\xi_i\big) = \mathbb{E}\big( \hat\Xi(G)\big) =  p\Big(1-(1-p)^{n-1}\Big)+\frac {1-p^n} n.
    $$
\end{thm}
\begin{proof}
    Let's do the same calculations as in the previous theorem, for $\hat\xi_i = \frac {\Big|E\big(\cN(i), V\setminus \cN(i)\big)\Big|} {d_i (n-d_i)}$. Note that by definition $\hat\xi_i = \frac 1 n$, when $d_i = 0$. Thus, 

    $$
        \mathbb{E}(\hat\xi_i) = \frac 1 n P(e = 0\,|\,d_i =0)\, P(d_i = 0)+\sumt_{k = 1}^{n-1}\sumt_{t = 0}^{f(k)} \frac {t+k} {k(n-k)}\, P(e = t+k\,|\,d_i = k)\, P(d_i = k) =
    $$
    $$
        = \frac {(1-p)^{n-1}} n+\sumt_{k = 1}^{n-1} C_{n-1}^k p^k (1-p)^{n-k-1}\sumt_{t = 0}^{f(k)} \frac {t+k} {k(n-k)}\,C_{f(k)}^t(1-p)^{f(k)-t} p^{t} =
    $$
    $$
        = \frac {(1-p)^{n-1}} n+\sumt_{k = 1}^{n-1} C_{n-1}^k p^k (1-p)^{n-k-1}\frac {k+p f(k)} {k(n-k)} = \frac {(1-p)^{n-1}} n+ \sumt_{k = 1}^{n-1} C_{n-1}^k p^k (1-p)^{n-1-k} \frac {1+p(n-1-k)} {n-k} =
    $$
    $$
          =\frac {(1-p)^{n-1}} n+ \sumt_{k = 1}^{n-1} C_{n-1}^k p^k (1-p)^{n-1-k} \frac {1+p(n-1-k)} {n-k} =\frac {(1-p)^{n-1}} n+ p-p(1-p)^{n-1}+
    $$
    $$
     + \sumt_{k = 1}^{n-1} C_{n-1}^k p^k (1-p)^{n-k}\frac 1 {n-k} = p-p(1-p)^{n-1}+ \sumt_{k = 0}^{n-1} \frac {(n-1)!} {(n-k)!\, k!} p^k (1-p)^{n-k} = 
    $$
    $$
    = p-p(1-p)^{n-1}+\frac 1 n\sumt_{k = 0}^{n-1} C_n^k p^k (1-p)^{n-k} =  p-p(1-p)^{n-1}+\frac {1-p^n} n.
    $$
    The same result for $\hat\Xi(G)$, since $\hat\Xi(G)$ is the average of  $\hat\xi_i$.
\end{proof}

If the number of vertices in the Erdos-Renyi graph $G(n,p)$ is large, then $\hat\Xi(G)\sim p$ and it is the same as the average clustering coefficient $C_{WS}(G)$. For a sparse Erdos-Renyi graph $G(n,p),\,p = \frac \l n$ the expected value of average normalized ksi-coefficient
$$
    \mathbb{E}\big(\hat\Xi(G)\big) = \frac \l n\Bigg(1-\Big(1-\frac \l n\Big)^{n-1}\Bigg)+\frac {1-\Big(\frac \l n\Big)^n} n = \frac {1+\l\big(1- e^{-\l}\big)} n + O\Big(\frac 1 {n^2}\Big),
$$

Therefore, it is asymptotically equivalent to the behavior of the expected value of average clustering coefficient $\mathbb{E}\big(C_{WS}(G)\big) = \frac \l n$. 

For fast calculations, the value $\Big|E\big(\cN(i), V\setminus \cN(i)\big)\Big|$ can be found by multiplying the adjacency matrix by two columns of the adjacency matrix:

\begin{lm}
$$
    E\big(\cN(i), V\setminus \cN(i)\big) = \sumt_{j, k\in V(G)} a_{ij} a_{jk} \overline a_{ki},
$$ where $\overline a_{ki} = 1-a_{ki}.$
\end{lm}
\begin{proof}
    Let's fix $i$ and note that
    $$\sumt_{j\in V(G)} a_{ij} a_{jk} = \begin{cases}
        d_i, & k = i, \\
        1, & i\sim j\sim k, \\
        0, & \text{otherwise},
    \end{cases}
    \qquad\text{and}\qquad
    1 - a_{ki} = \begin{cases}
        1, & k = i, \\
        1, & k\not\sim i, \\
        0, & k\sim i.
    \end{cases}
    $$
    Therefore, 
    $$
        \Big|E\big(\cN(i), V\setminus \cN(i)\big)\Big| = d_i+ \big|k, j\in V(G): i\sim j\sim k, k\not\sim i\big| =  \sumt_{j, k\in V(G)} a_{ij} a_{jk} \overline a_{ki}.
    $$
\end{proof}

\begin{cor}\label{cor1} For each vertex $i$
    $$
        \xi_i = \frac {\Big(A^2\cdot\overline A\Big)_{ii}} {\big(A^2\Big)_{ii}},
    $$
    where $\overline{A} = I-A$ for $I$ --- matrix of all ones.
\end{cor}

Therefore, the calculation complexity for ksi-centrality is the same as for the local clustering centrality $c_i = \frac {\big(A^3\big)_{ii}} {\big(A^2\big)_{ii}\Big(\big(A^2\big)_{ii}-1\Big)}.$

It is easy to see that by definition, the centrality ksi and the normalized centrality ksi are related to the Chegeer's number: 

\begin{thm}\label{thm:nksi1} Let $h(G)$ be the Chegeer's number of $G$.
\begin{enumerate}
    \item If the degree $d_i\leq\frac n 2$ for a vertex $i$, then $\xi_i\geq h(G)$,
    \item $\hat\xi_i\geq \begin{cases}
            h(G)\, (n-d_i), & \text{if }\, d_i\leq\frac n 2, \\
            h(G)\, d_i, & \text{otherwise.}
    \end{cases}$
\end{enumerate}

\end{thm}

The next result is that the normalized ksi-centrality (as well as the average normalized ksi-coefficient) are bounded by the $\l_2$ algebraic connectivity (or the second eigenvalue of the Laplacian matrix).

\begin{thm}\label{thm:nksi2}
    Consider connected graph $G$. Let's denote Laplacian matrix spectrum by $0=\l_1<\l_2\leq \l_3\leq\ ...\ \leq \l_n$. For any vertex $i\in V(G)$
    $$
   \hat\xi_i\geq  \frac {\l_2} n, \qquad
   \hat\Xi(G)\geq  \frac {\l_2} n.
    $$
\end{thm}
\begin{proof}
    Remind that $\l_2 = \mint_{x\in \rR^n, (x,\bf{1}) = 0} \frac {(Lx, x)} {(x,x)} = \mint_{x\in \rR^n, (x,\bf{1}) = 0} \frac {\sumt_{i,j\in V(G), i\sim j}(x_i-x_j)^2} {(x,x)},$ where \bf{1} is the vector of ones and $(\cdot,\cdot)$ is the standard dot product in $\rR^n$. Let's define for each vertex $i$ a vector
    $$
        y = \big(y_j\big) = \begin{cases}
            n-d_i, & j\in \cN(i), \\
            -d_i, & j\in V(G)\setminus\cN(i).
        \end{cases}
    $$
    Therefore, $(y, \bf{1}) = 0$ and $(y,y) = (n-d_i)^2d_i+d_i^2(n-d_i) = d_i(n-d_i)n$. Also $\sumt_{k,j\in V(G), k\sim j}(y_k-y_j)^2 = n^2 \Big|E\big(\cN(i), V\setminus\cN(i)\big)\Big|.$ Therefore,
    $$
        \l_2\leq \frac {\sumt_{k,j\in V(G), k\sim j}(y_i-y_j)^2} {(y,y)} = n \frac {\Big|E\big(\cN(i), V\setminus\cN(i)\big)\Big|} {d_i(n-d_i)} = n\,\hat \xi_i.
    $$
\end{proof}

For a star graph, it is easy to calculate the ksi-centrality and normalized ksi-centrality with their average values.

\begin{thm}\label{thm:star}
For the star graph with $n+1$ vertices
    $$
        \hat\xi_i = 1, \qquad \xi_i = \begin{cases}
            1, & \text{if $i$ central vertex,} \\
            n, & \text{otherwise,}
        \end{cases}
    $$
    $$
        \hat\Xi(G) = 1,\qquad \Xi(G) = \frac {n^2+1} {n+1}\sim n.
    $$
\end{thm}   
    According to this theorem, the normalized ksi-centrality does not distinguish between the vertices of a star graph (as does the local clustering coefficient), and its values are equal to the value of an isolated vertex. However, for ksi-centrality, peripheral nodes are more important because they have a significant central neighbor. Furthermore, the average normalized ksi-coefficient is constant but the average ksi-coefficient tends to infinity as the number of nodes increases.

\begin{figure}[h!]
    \centering
	\includegraphics[width = 1\textwidth]{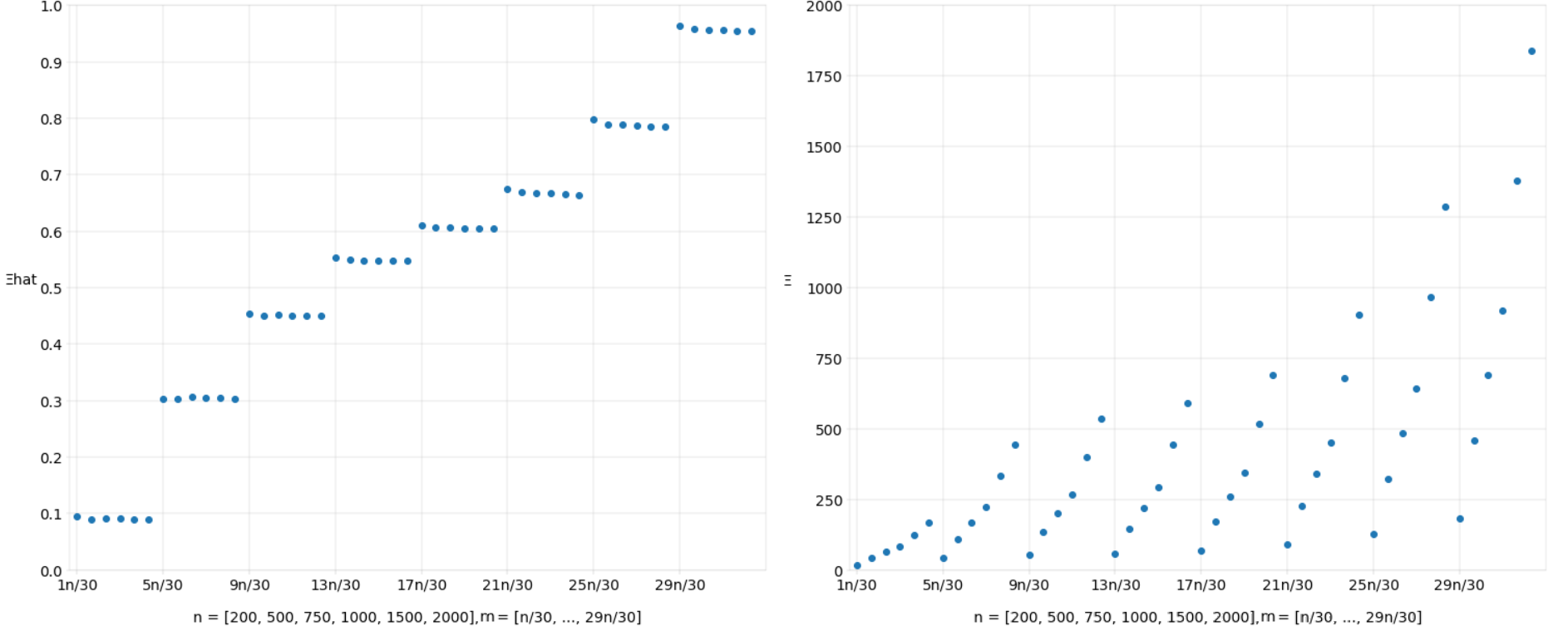}
	\caption{Normalized average ksi (left) and average ksi (right) coefficients for different parameters of Barabasi-Albert model. The average coefficients are combined into groups with an equal ratio $m/n$ for $n = 200, 500, 750, 1000, 1500, 2000$ and plotted with the same step for the corresponding value of the ratio.}
	\label{fig:ksirel}
\end{figure}

\begin{figure}[h!]
    \centering
	\includegraphics[height = 0.8\textheight, width = 1.0\textwidth]{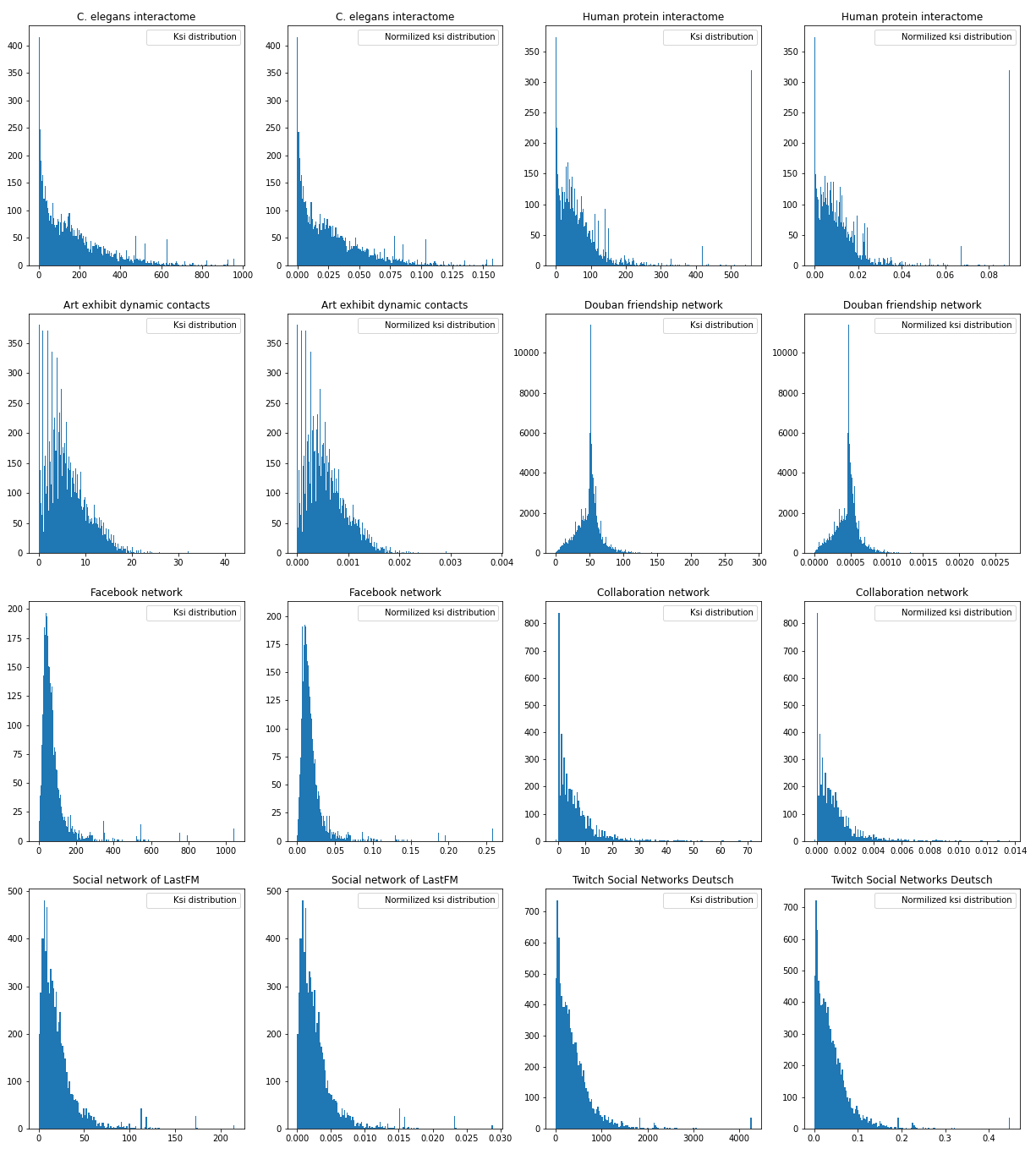}
	\caption{Ksi (left) and normalized ksi  (right) distributions for different networks. The x-axis shows the ksi values, and the y-axis shows the number of vertices.}
	\label{fig:ksiVSksi}
\end{figure}

\begin{table}[h!]
\resizebox{\textwidth}{!}{
\begin{tabular}{|c|c|}
     \hline
     Network name & Internet link \\
     \hline
    C. elegans interactome	\cite{bib1} &	\url{https://networks.skewed.de/net/celegans_interactomes} \\
Human protein interactome	\cite{bib2}&	\url{https://networks.skewed.de/net/reactome}	\\
Art exhibit dynamic contacts \cite{bib3}	&	\url{https://networks.skewed.de/net/sp_infectious}	\\
Douban friendship network	\cite{bib4}&	\url{https://networks.skewed.de/net/douban}	\\
Facebook network	\cite{bib5}&	\url{https://snap.stanford.edu/data/ego-Facebook.html}	\\
Collaboration network	\cite{bib6}&	\url{https://snap.stanford.edu/data/ca-GrQc.html}	\\
Social network of LastFM	\cite{bib7}&	\url{https://snap.stanford.edu/data/feather-lastfm-social.html}	\\
Twitch Social Networks Deutsch	\cite{bib8}&	\url{https://snap.stanford.edu/data/twitch-social-networks.html}	\\
Infrastructure Network openflights	\cite{bib9}&	\url{https://networkrepository.com/inf-openflights.php}	\\
Social Network advogato	\cite{bib9}&	\url{https://networkrepository.com/soc-advogato.php}	\\
Web Graph EPA	\cite{bib9}&	\url{https://networkrepository.com/web-EPA.php}	\\
Web Graph spam	\cite{bib9}&	\url{https://networkrepository.com/web-spam.php}	\\
Gemsec Deezer dataset Croatia 	\cite{bib10}&	\url{https://snap.stanford.edu/data/gemsec-Deezer.html}	\\
Gemsec Deezer dataset Hungary	\cite{bib10}&	\url{https://snap.stanford.edu/data/gemsec-Deezer.html}	\\
Gemsec Deezer dataset Romania	\cite{bib10}&	\url{https://snap.stanford.edu/data/gemsec-Deezer.html}	\\
Gemsec Facebook dataset artist	\cite{bib10}&	\url{https://snap.stanford.edu/data/gemsec-Facebook.html}	\\
Gemsec Facebook dataset athletes	\cite{bib10}&	\url{https://snap.stanford.edu/data/gemsec-Facebook.html}	\\
Gemsec Facebook dataset company	\cite{bib10}&	\url{https://snap.stanford.edu/data/gemsec-Facebook.html}	\\
Gemsec Facebook dataset government	\cite{bib10}&	\url{https://snap.stanford.edu/data/gemsec-Facebook.html}	\\
Gemsec Facebook dataset new sites	\cite{bib10}&	\url{https://snap.stanford.edu/data/gemsec-Facebook.html}	\\
Gemsec Facebook dataset politician	\cite{bib10}&	\url{https://snap.stanford.edu/data/gemsec-Facebook.html}	\\
Gemsec Facebook dataset public figure	\cite{bib10}&	\url{https://snap.stanford.edu/data/gemsec-Facebook.html	}\\
Gemsec Facebook dataset tvshow	\cite{bib10}&	\url{https://snap.stanford.edu/data/gemsec-Facebook.html}	\\
DBLP collaboration network	\cite{bib11}&	\url{https://snap.stanford.edu/data/com-DBLP.html}	\\
Gowalla location based online social	\cite{bib12}&	\url{https://snap.stanford.edu/data/loc-Gowalla.html}	\\
Brightkite location based online social	\cite{bib12}&	\url{https://snap.stanford.edu/data/loc-Brightkite.html}	\\
Amazon product network	\cite{bib11}&	\url{https://snap.stanford.edu/data/com-Amazon.html} \\
Email communication from Enron	\cite{bib13}&	\url{https://snap.stanford.edu/data/email-Enron.html}	\\
Arxiv High Energy paper citation	\cite{bib14}&	\url{https://snap.stanford.edu/data/cit-HepPh.html}	\\
Biological Network grid-human	\cite{bib9}&	\url{https://networkrepository.com/bio-grid-human.php}	\\
Vidal human interactome	\cite{bib15}&	\url{https://networks.skewed.de/net/interactome_vidal}	\\
Brain Network fly-drosophila-medulla	\cite{bib9}&	\url{https://networkrepository.com/bn-fly-drosophila-medulla-1.php}	\\
Protein interactomes across tree of life	\cite{bib16}&	\url{https://networks.skewed.de/net/tree-of-life}	\\
Marvel Universe social network	\cite{bib17}&	\url{https://networks.skewed.de/net/marvel_universe}	\\
Global nematode-mammal interactions	\cite{bib18}&	\url{https://networks.skewed.de/net/nematode_mammal}	\\
Internet Autonomous Systems graph	\cite{bib19}&	\url{https://networks.skewed.de/net/topology} \\
Scientific collaborations in physics	\cite{bib20}&	\url{https://networks.skewed.de/net/arxiv_collab}	\\
EU national procurement FR 2011	\cite{bib21}&	\url{https://networks.skewed.de/net/eu_procurements_alt}	\\
Stack Overflow favorites	\cite{bib22}&	\url{https://networks.skewed.de/net/stackoverflow}	\\
WordNet relationships	\cite{bib23}&	\url{https://networks.skewed.de/net/wordnet}	\\
\hline
\end{tabular}}

\caption{List of used networks with references.}
        \label{tab:net}
    \end{table}

\begin{figure}[h!]
    \centering
	\includegraphics[height = 0.8\textheight, width = 1\textwidth]{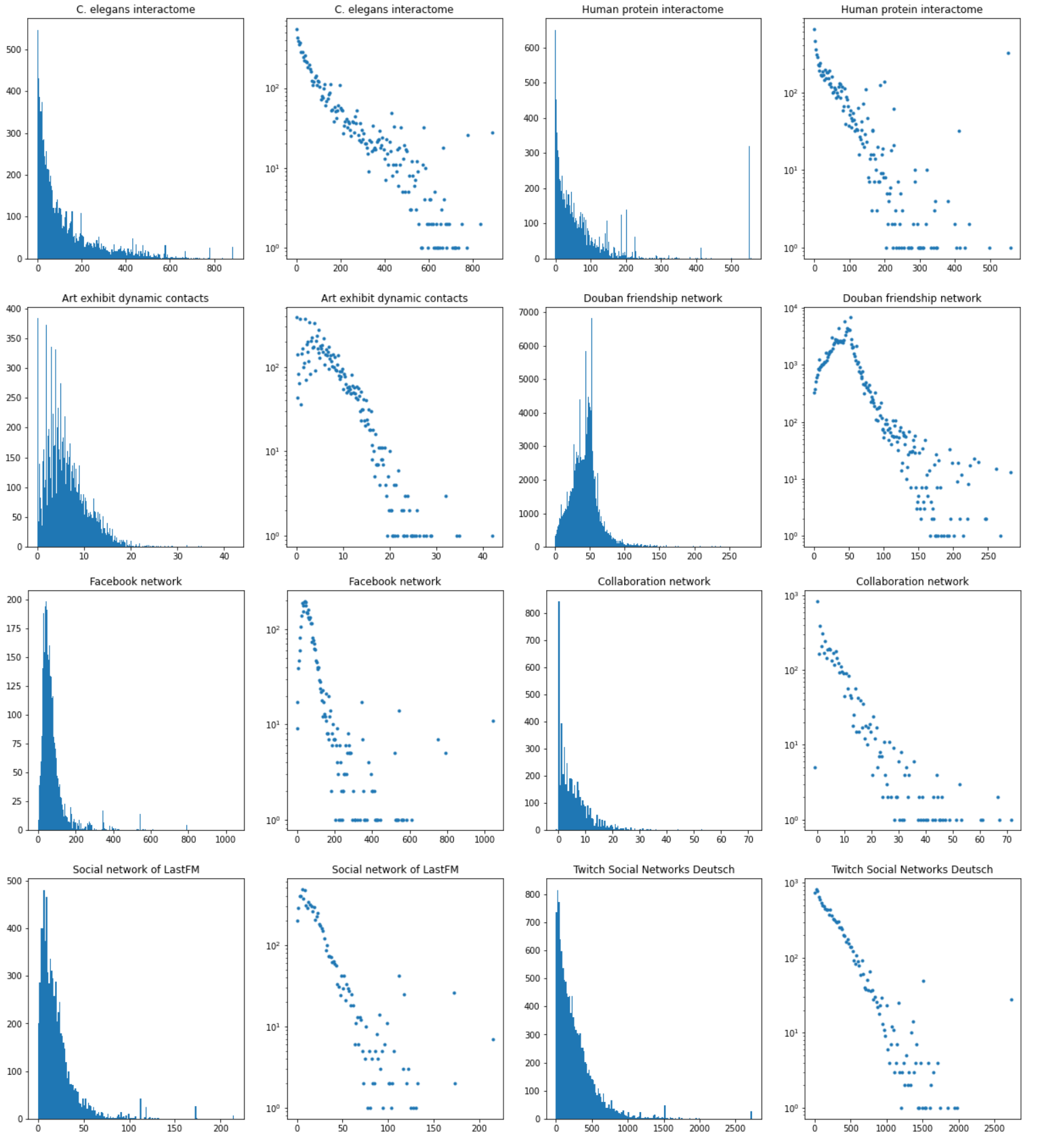}
	\caption{Ksi-distributions (left) and corresponding log-scale for y axes (right) for networks from table~\ref{tab:net}. The x-axis shows the ksi values, and the y-axis shows the number of vertices.}
	\label{fig:1}
\end{figure}

\begin{figure}[h!]
    \centering
	\includegraphics[height = 0.8\textheight, width = 1\textwidth]{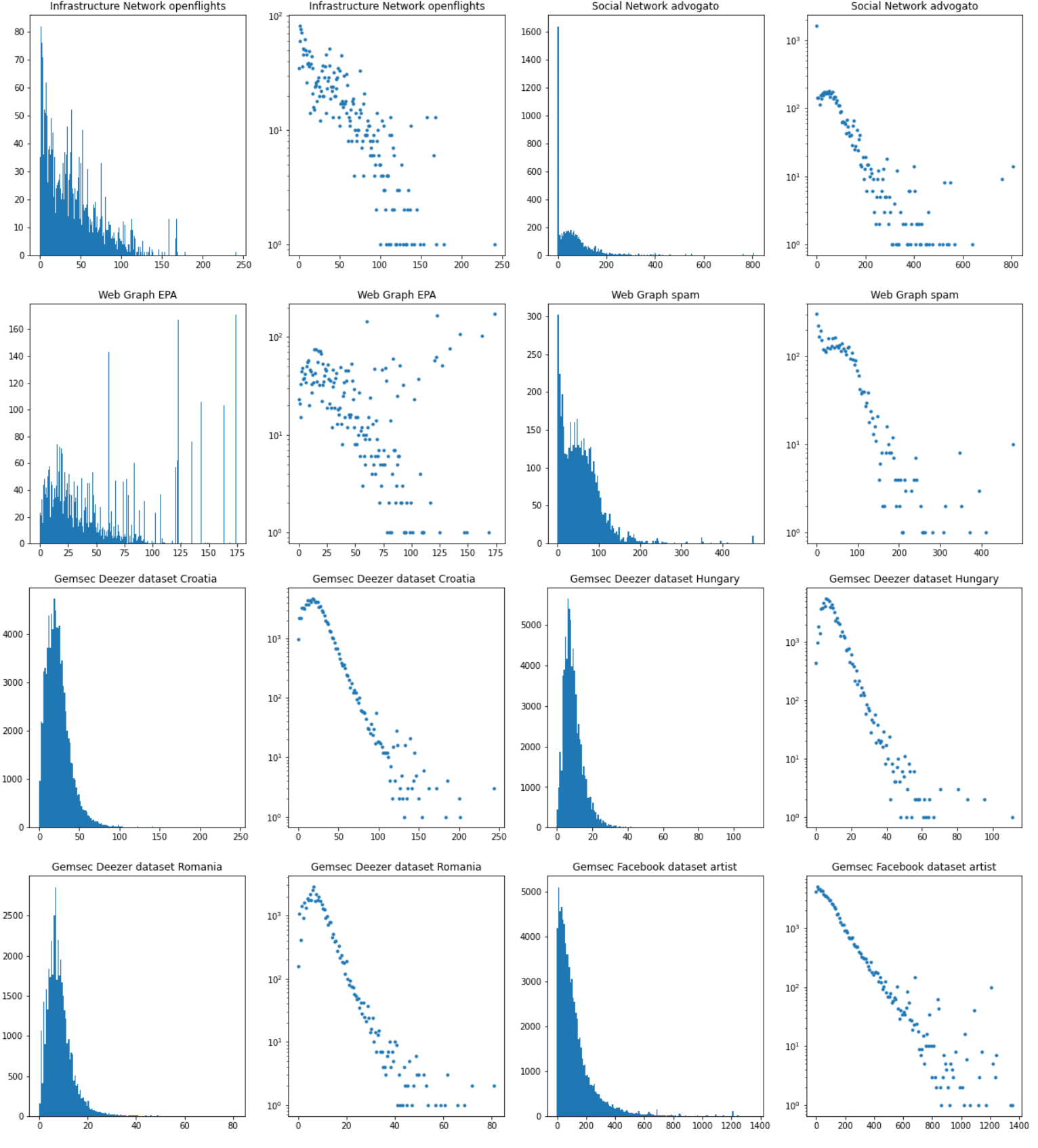}
	\caption{Ksi-distributions (left) and corresponding log-scale for y axes (right) for networks from table~\ref{tab:net}. The x-axis shows the ksi values, and the y-axis shows the number of vertices.}
	\label{fig:2}
\end{figure}

\begin{figure}[h!]
    \centering
	\includegraphics[height = 0.8\textheight, width = 1\textwidth]{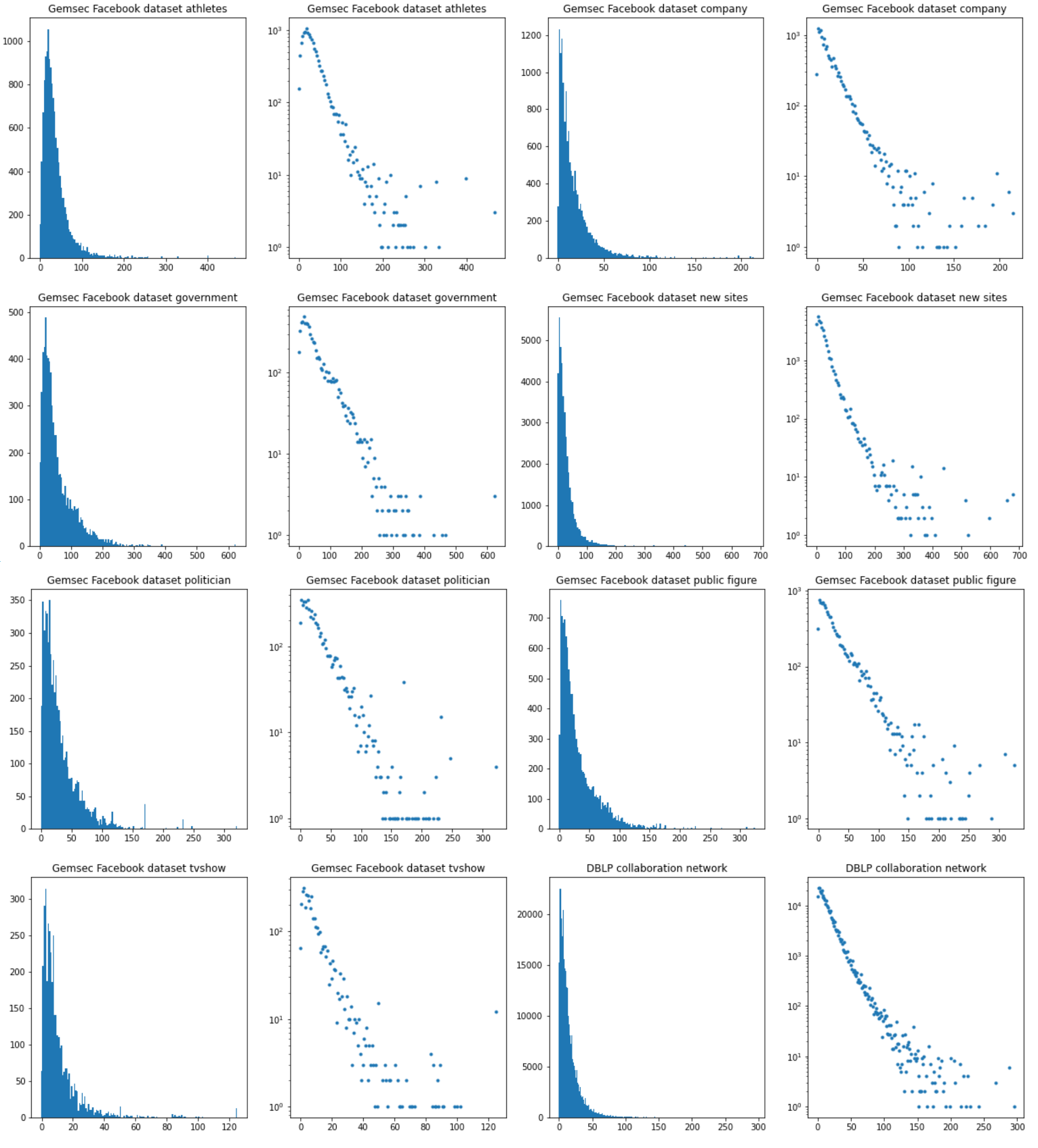}
	\caption{Ksi-distributions (left) and corresponding log-scale for y axes (right) for networks from table~\ref{tab:net}. The x-axis shows the ksi values, and the y-axis shows the number of vertices.}
	\label{fig:3}
\end{figure}

\begin{figure}[h!]
    \centering
	\includegraphics[height = 0.8\textheight, width = 1\textwidth]{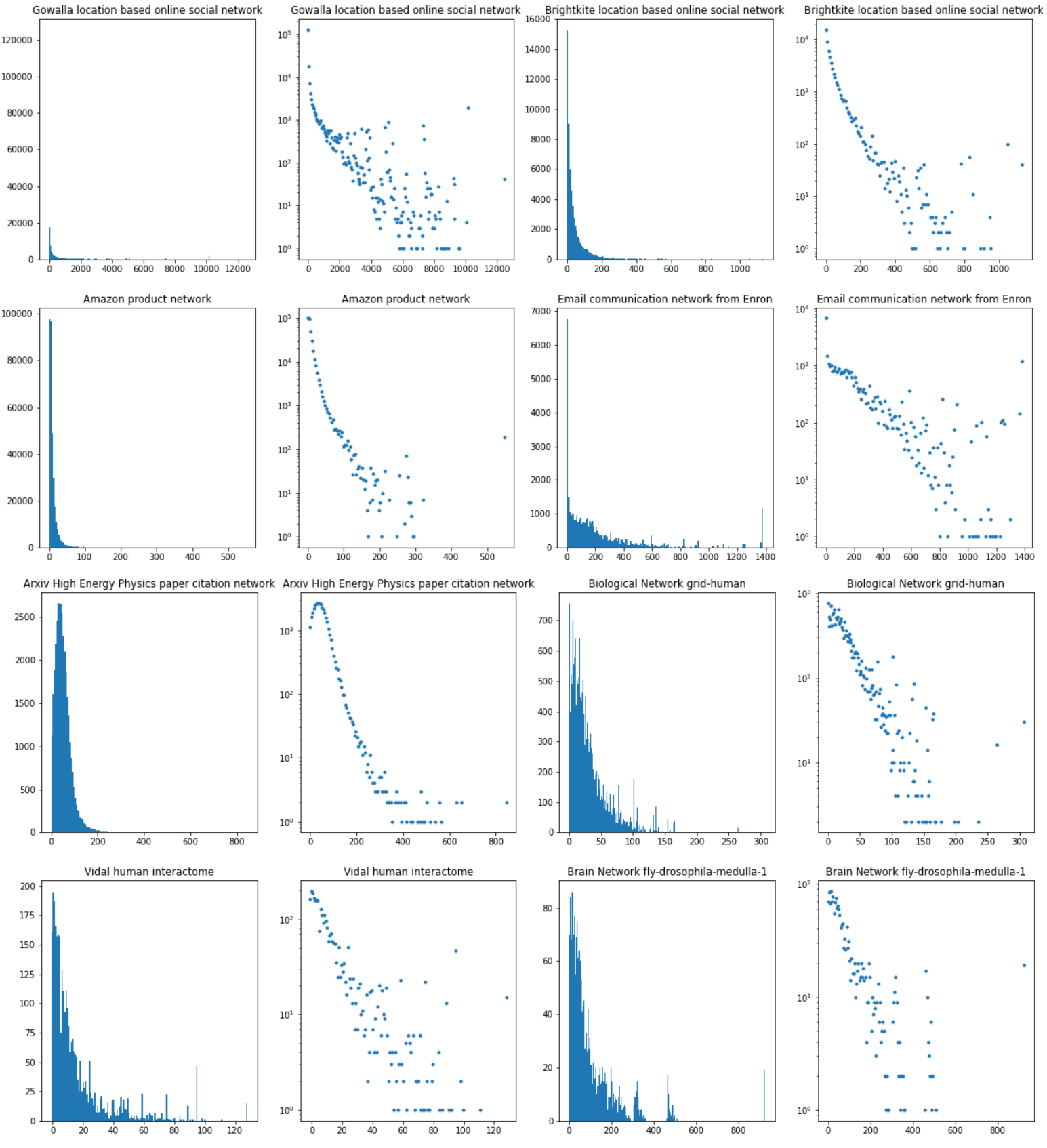}
	\caption{Ksi-distributions (left) and corresponding log-scale for y axes (right) for networks from table~\ref{tab:net}. The x-axis shows the ksi values, and the y-axis shows the number of vertices.}
	\label{fig:4}
\end{figure}

\begin{figure}[h!]
    \centering
	\includegraphics[height = 0.8\textheight, width = 1\textwidth]{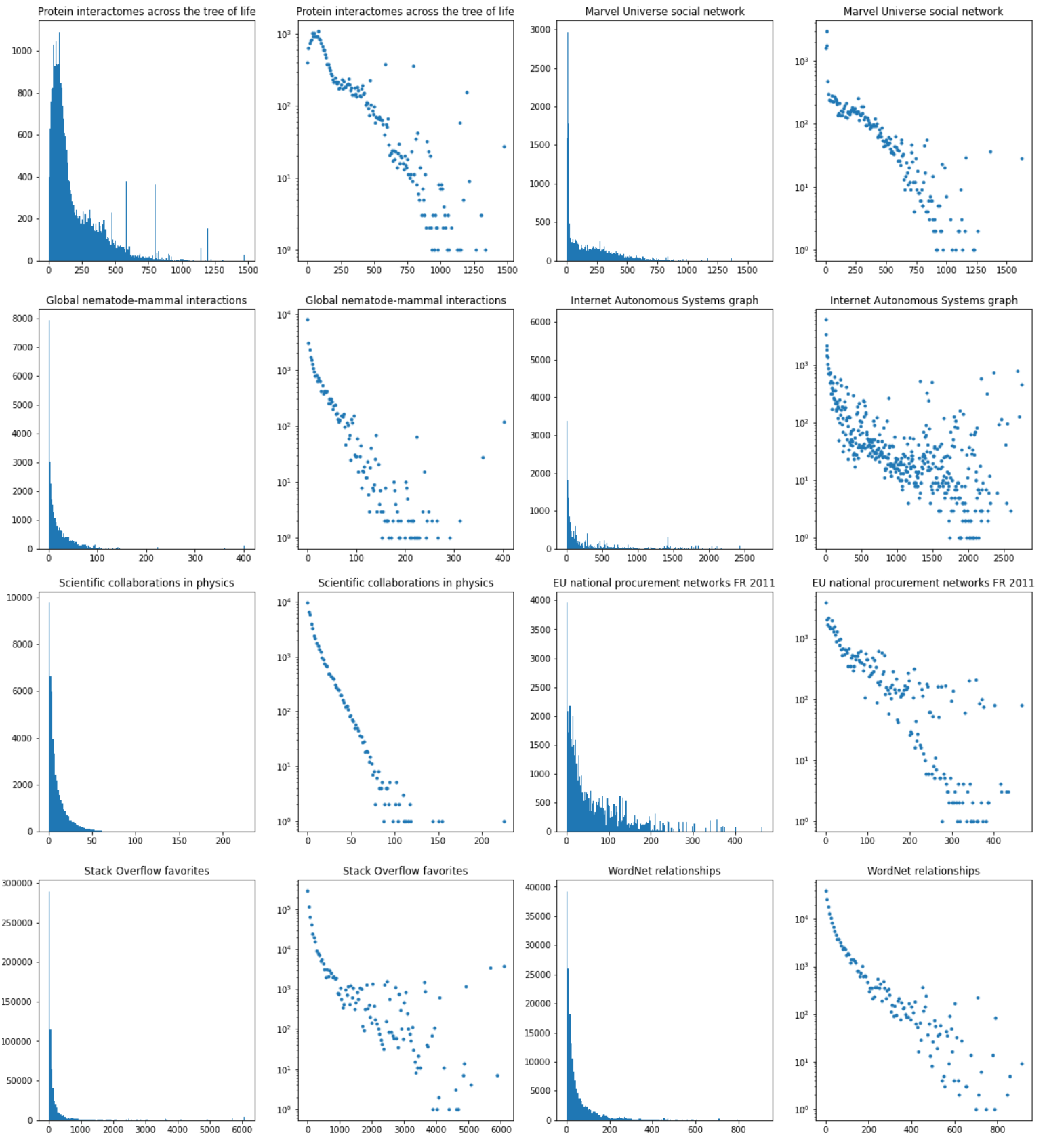}
	\caption{Ksi-distributions (left) and corresponding log-scale for y axes (right) for networks from table~\ref{tab:net}. The x-axis shows the ksi values, and the y-axis shows the number of vertices.}
	\label{fig:5}
\end{figure}

\begin{figure}[h!]
    \centering
	\includegraphics[height = 0.8\textheight, width = 1\textwidth]{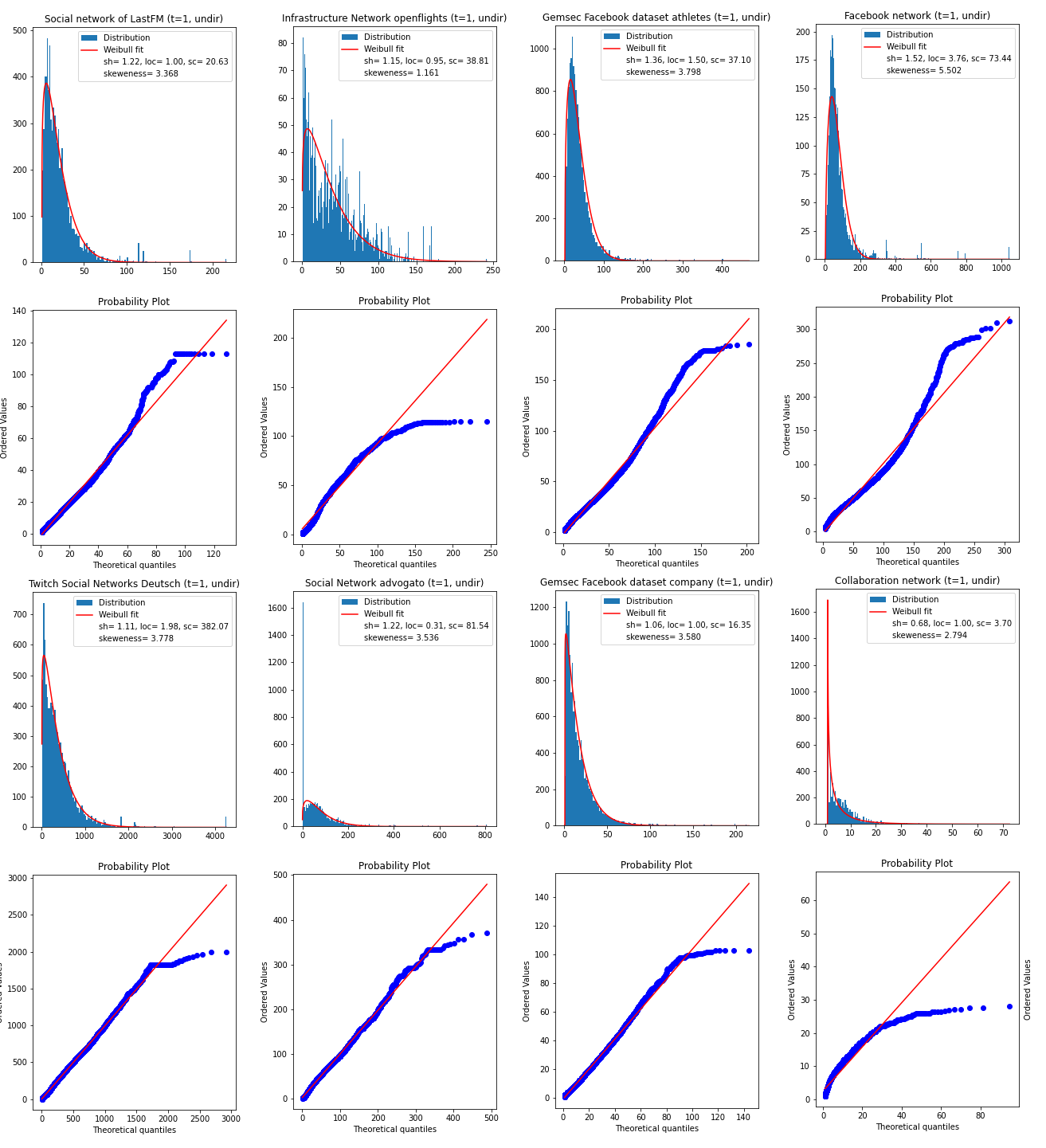}
	\caption{Ksi-distributions fitted with Weibull distribution, their Pearson's moment coefficients of skewness (top) and corresponded QQ-plots (bottom) for networks from the table~\ref{tab:net}. The x-axis shows the ksi values, and the y-axis shows the number of vertices for distributions.}
	\label{fig:wb}
\end{figure}

\begin{figure}[h!]
    \centering
	\includegraphics[width = 1\textwidth]{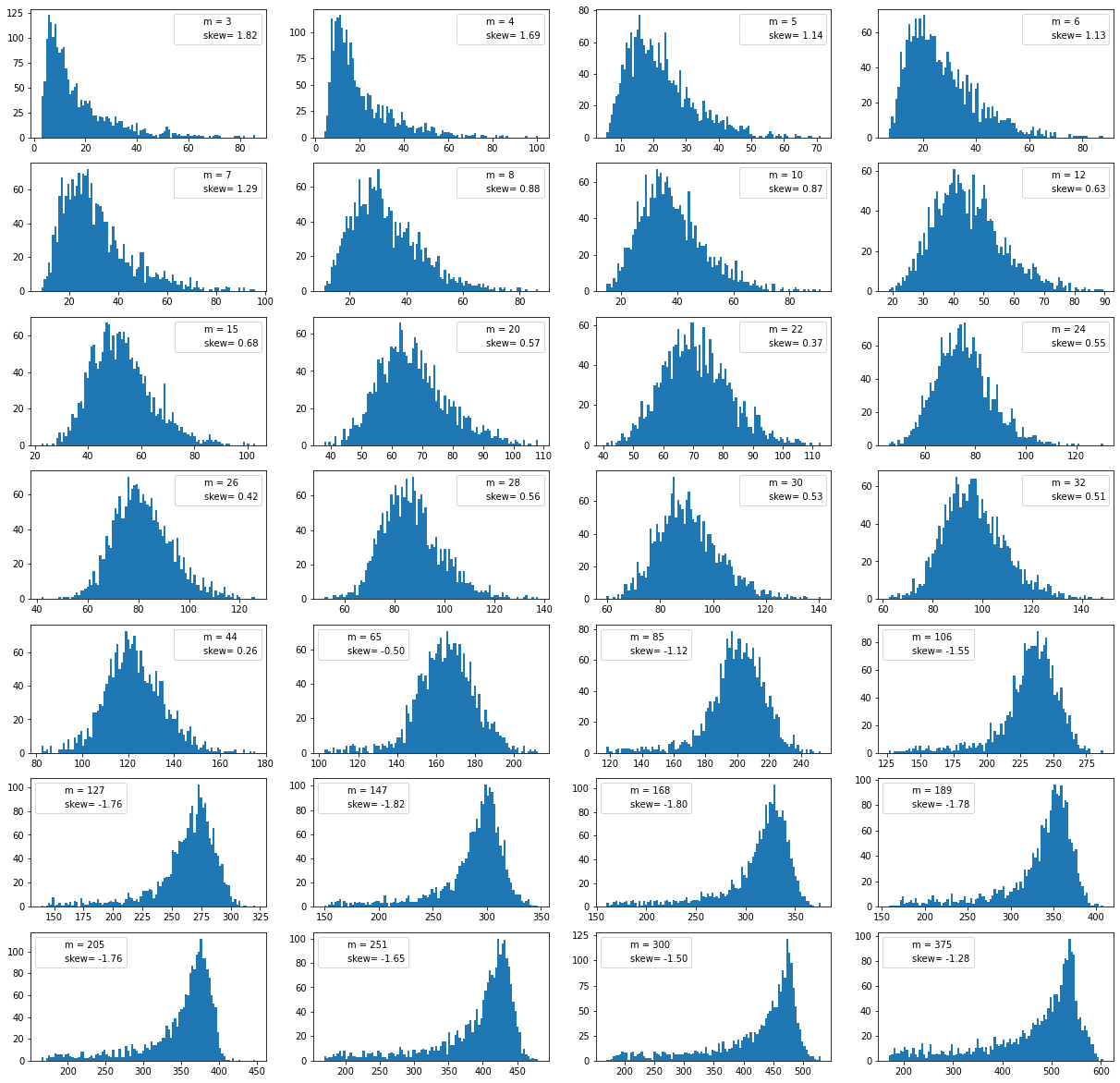}
	\caption{Dependence of the ksi-distribution and its Pearson’s moment coefficients of skewness on the parameter $m$ for the Barabasi-Albert model $n=2000$. The x-axis shows the ksi values, and the y-axis shows the number of vertices for distributions.}
	\label{fig:BA}
\end{figure}

\begin{figure}[h!]
    \centering
	\includegraphics[width = 1\textwidth]{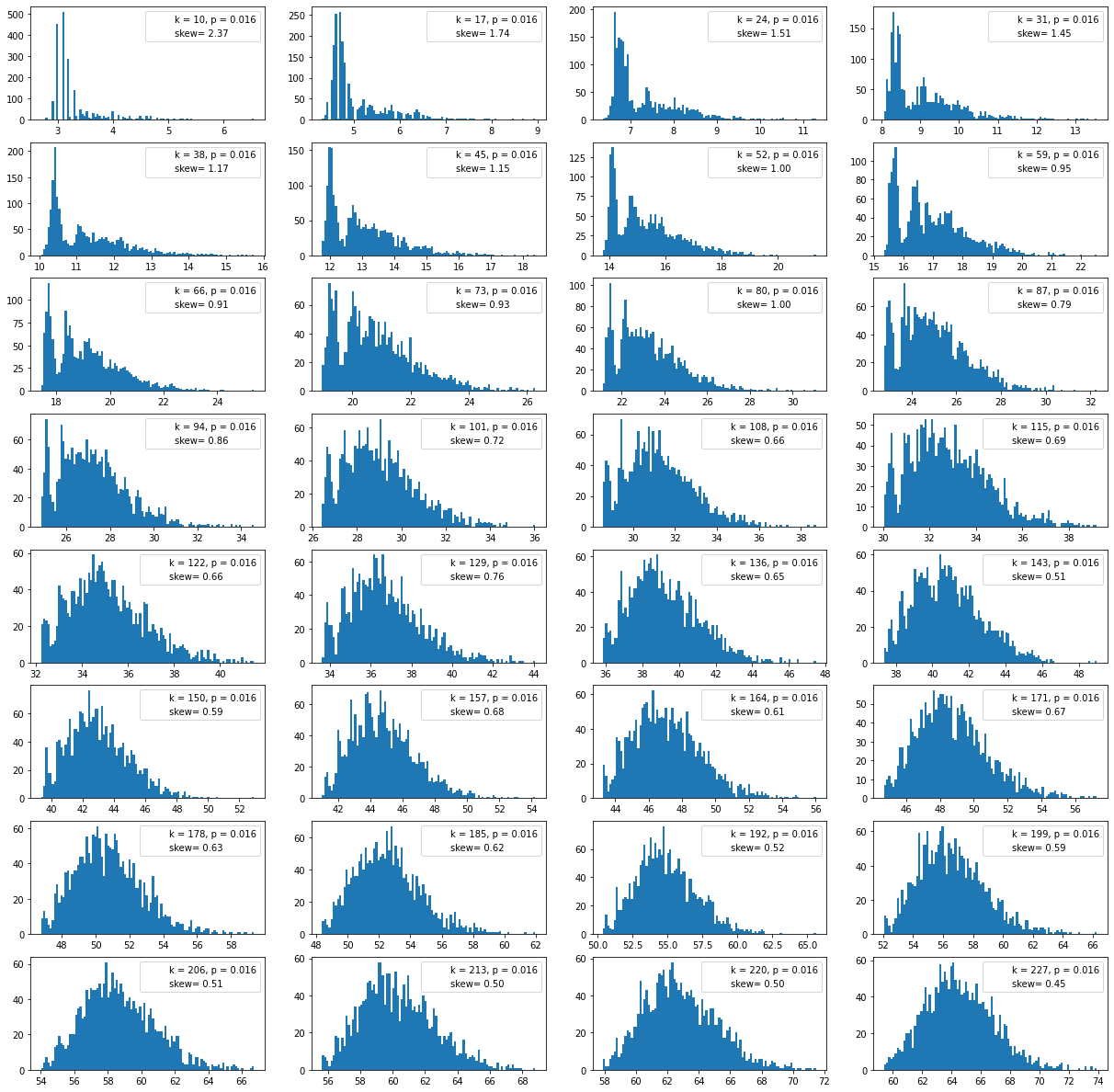}
	\caption{Dependence of the ksi-distribution and its Pearson’s moment coefficients of skewness on the parameter $k$ for the Watts-Strogatz model $n=2000, p=0.016$. The x-axis shows the ksi values, and the y-axis shows the number of vertices for distributions.}
	\label{fig:WS1}
\end{figure}

\begin{figure}[h!]
    \centering
	\includegraphics[width = 1\textwidth]{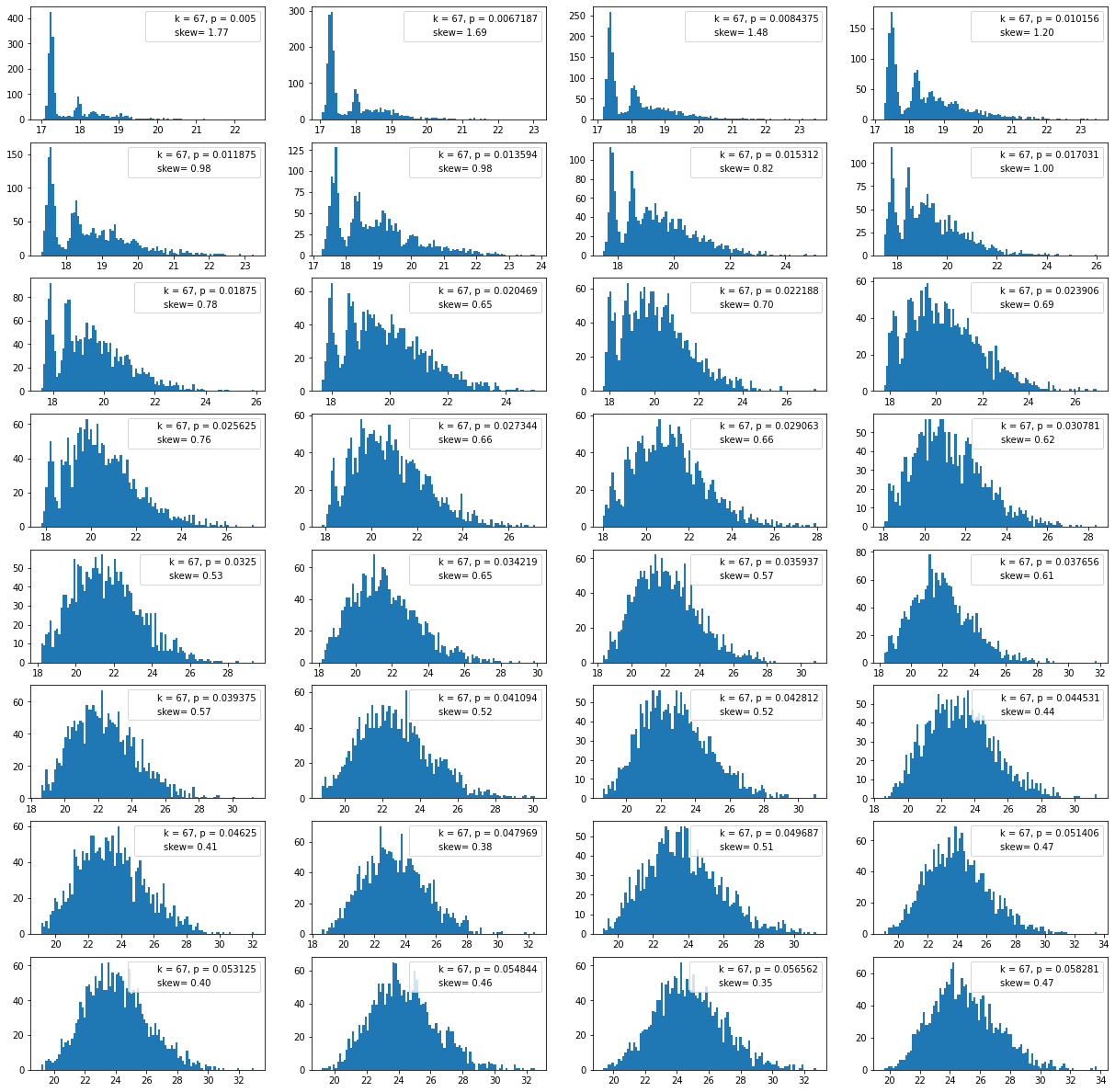}
	\caption{Dependence of the ksi-distribution and its Pearson’s moment coefficients of skewness on the parameter $p$ for the Watts-Strogatz model $n=2000,k=67$. The x-axis shows the ksi values, and the y-axis shows the number of vertices for distributions.}
	\label{fig:WS2}
\end{figure}

\begin{figure}[h!]
    \centering
	\includegraphics[width = 0.6\textwidth]{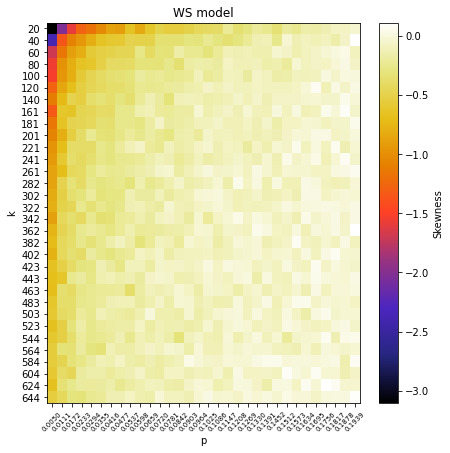}
	\caption{Pearson’s moment coefficients of skewness for the ksi-distribution for the Watts-Strogatz model $n=2000, k, p$.}
	\label{fig:WS2}
\end{figure}

\begin{figure}[h!]
    \centering
	\includegraphics[width = 1\textwidth]{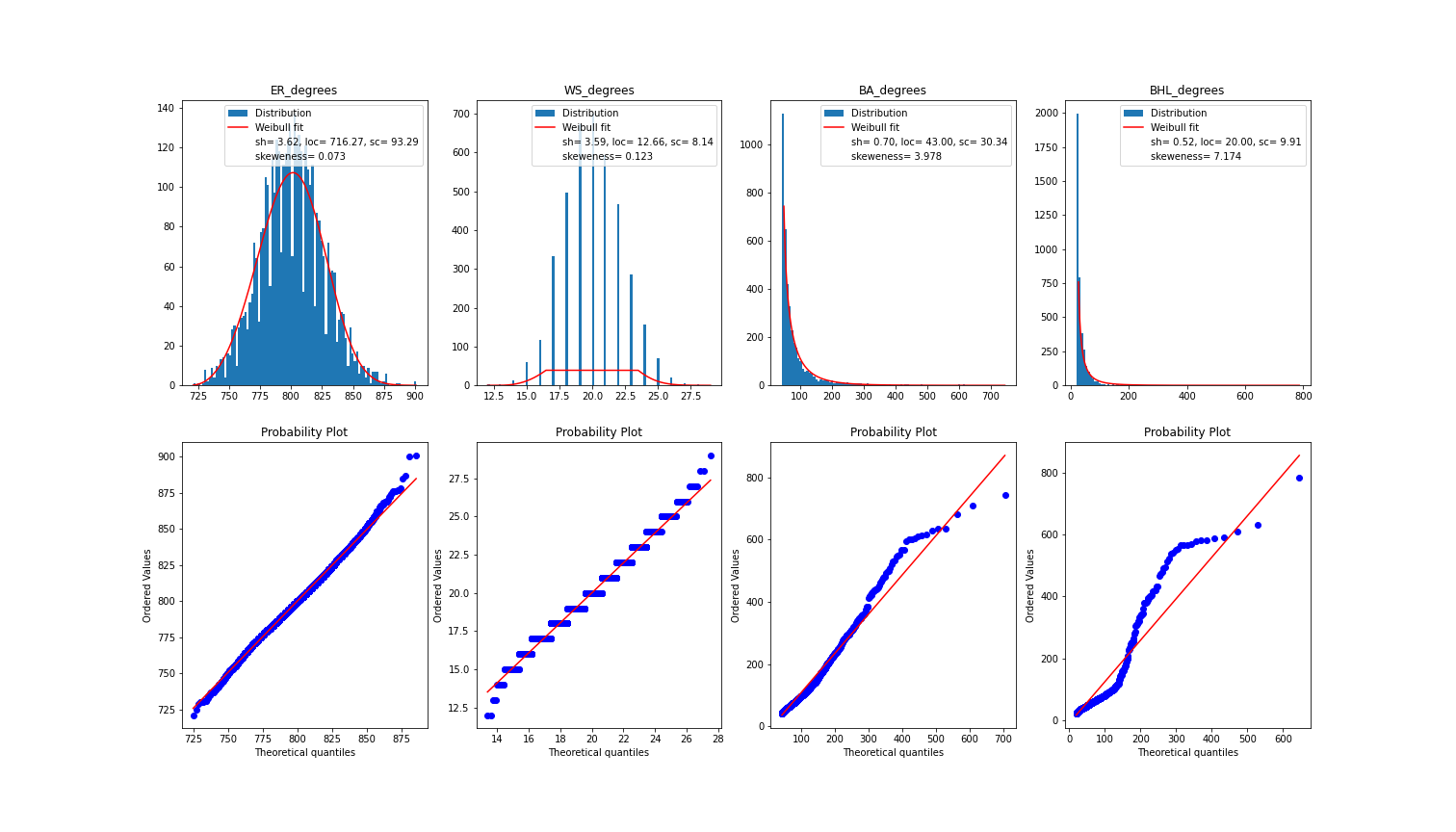}
	\caption{Degree distributions fitted with Weibull distribution, their Pearson's moment coefficients of skewness (top) and corresponded QQ-plots (bottom) for networks from models Erdos-Renyi $(n=4000, p=0.2)$, Watts-Strogatz $(n=4000, k=20, p=0.3)$, Barabasi-Albert $(n=4000, m=43)$ and Boccaletti-Hwang-Latora $(n_0=100,  m=20, n=4000)$. The x-axis shows degree values, and the y-axis shows the number of vertices for distributions.}
	\label{fig:weib_degs}
\end{figure}

\begin{figure}[h!]
    \centering
	\includegraphics[height = 0.8\textheight, width = 1\textwidth]{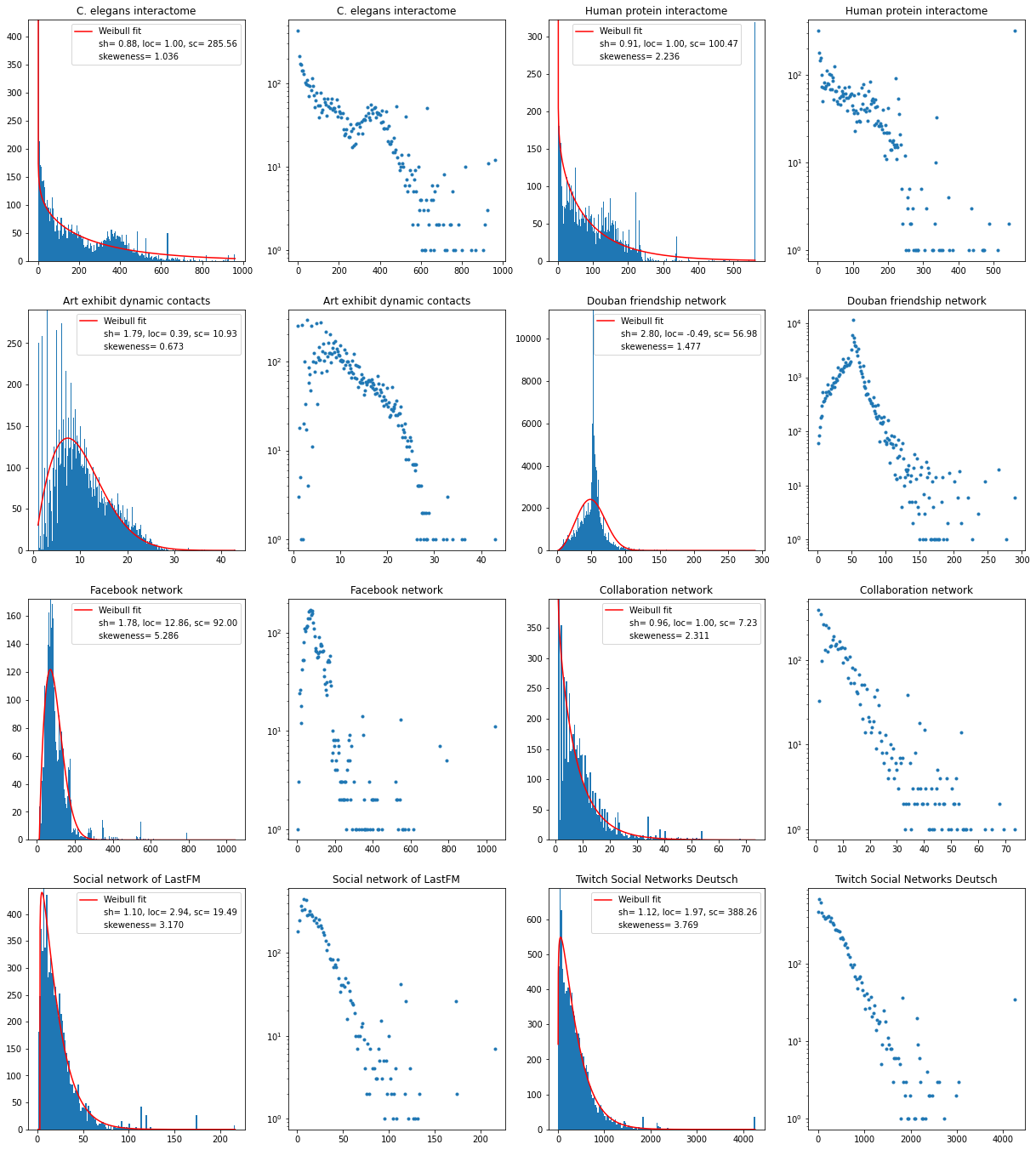}
	\caption{Average neighbor degree distributions (left) and corresponding log-scale for y axes (right) for networks from table~\ref{tab:net}. The x-axis shows the average neighbor degree values, and the y-axis shows the number of vertices.}
	\label{fig:ad1}
\end{figure}

\begin{figure}[h!]
    \centering
	\includegraphics[height = 0.8\textheight, width = 1\textwidth]{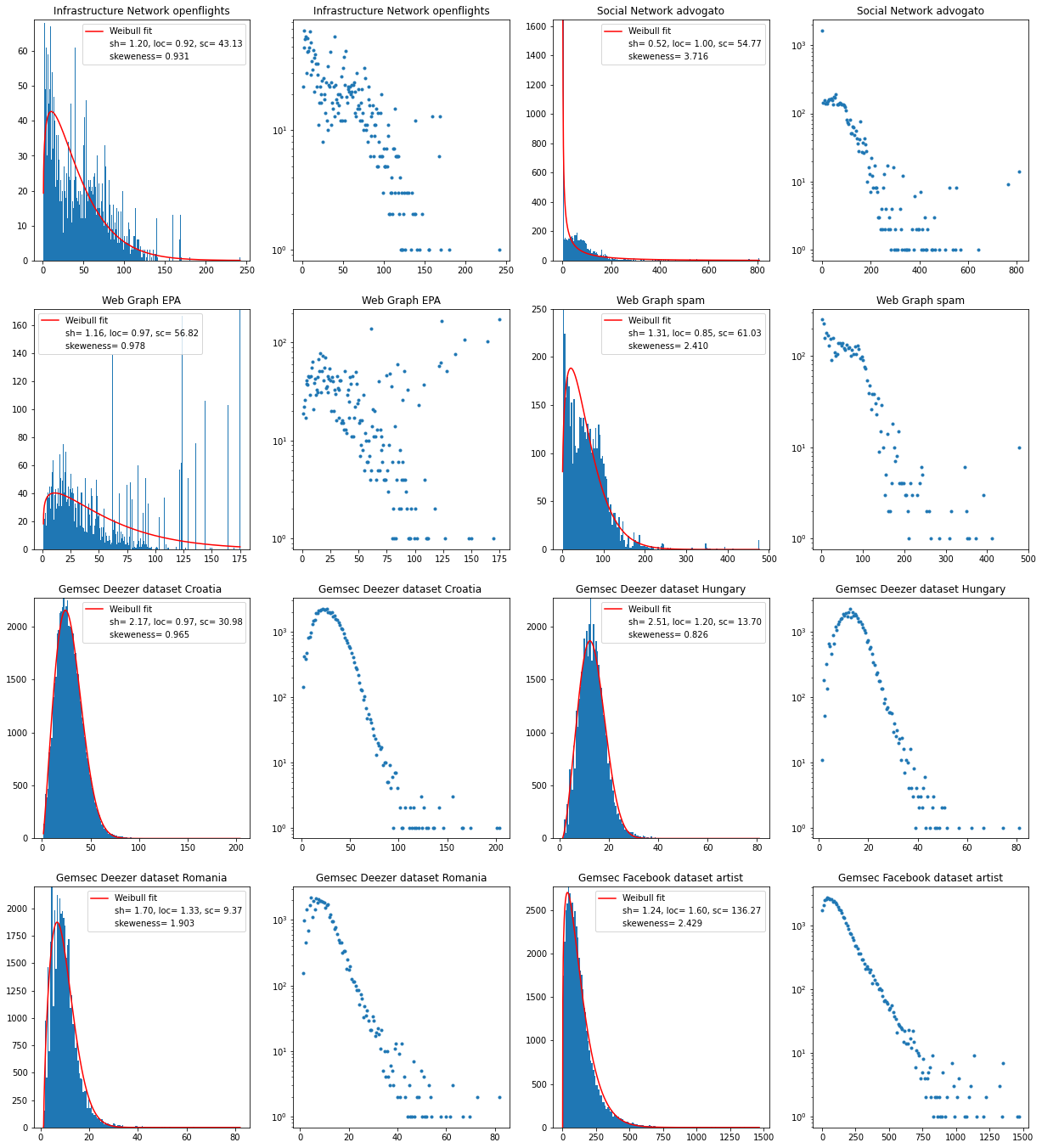}
	\caption{Average neighbor degree distributions (left) and corresponding log-scale for y axes (right) for networks from table~\ref{tab:net}. The x-axis shows the average neighbor degree values, and the y-axis shows the number of vertices.}
\end{figure}

\begin{figure}[h!]
    \centering
	\includegraphics[height = 0.8\textheight, width = 1\textwidth]{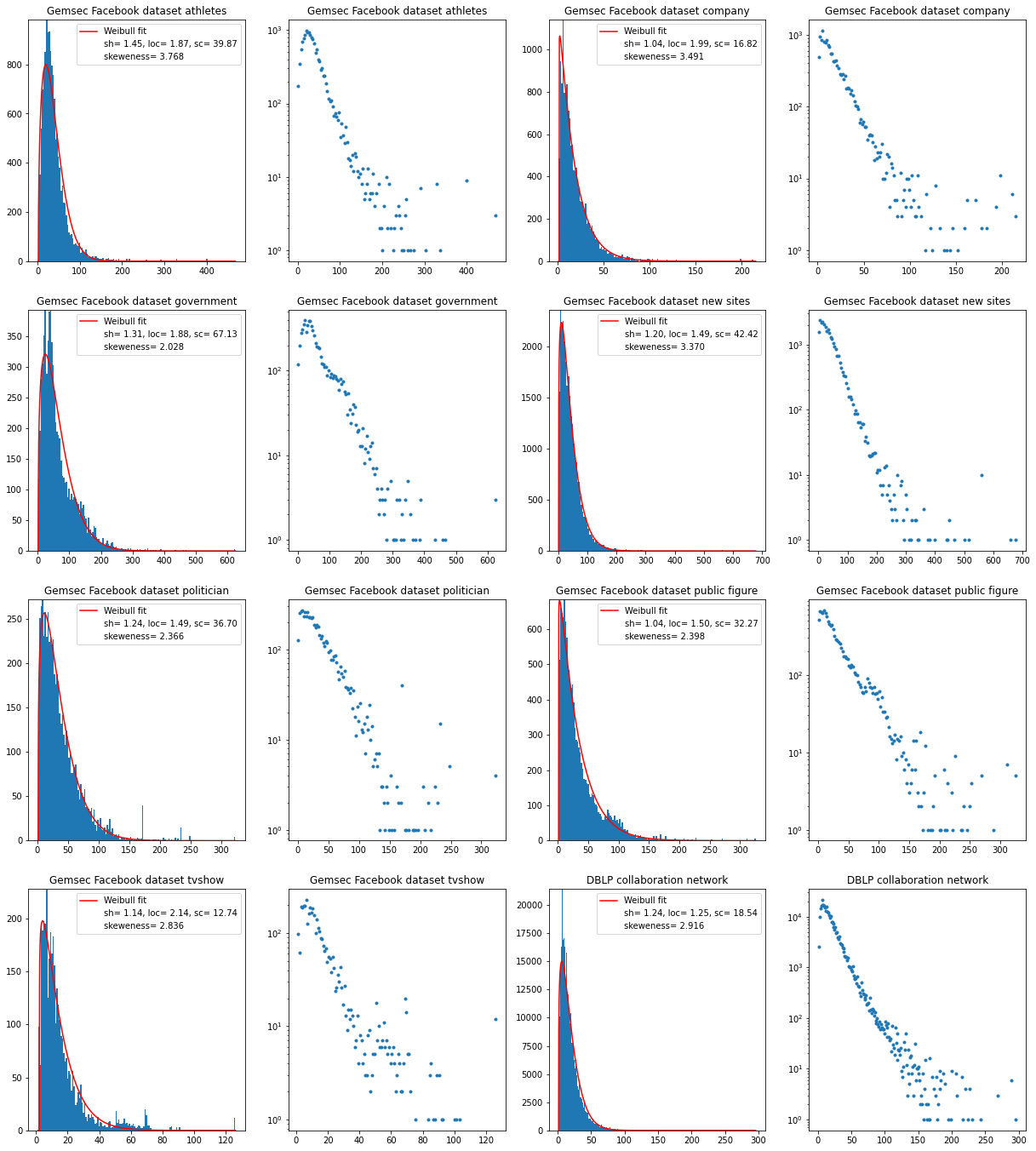}
	\caption{Average neighbor degree distributions (left) and corresponding log-scale for y axes (right) for networks from table~\ref{tab:net}. The x-axis shows the average neighbor degree values, and the y-axis shows the number of vertices.}
\end{figure}

\begin{figure}[h!]
    \centering
	\includegraphics[height = 0.8\textheight, width = 1\textwidth]{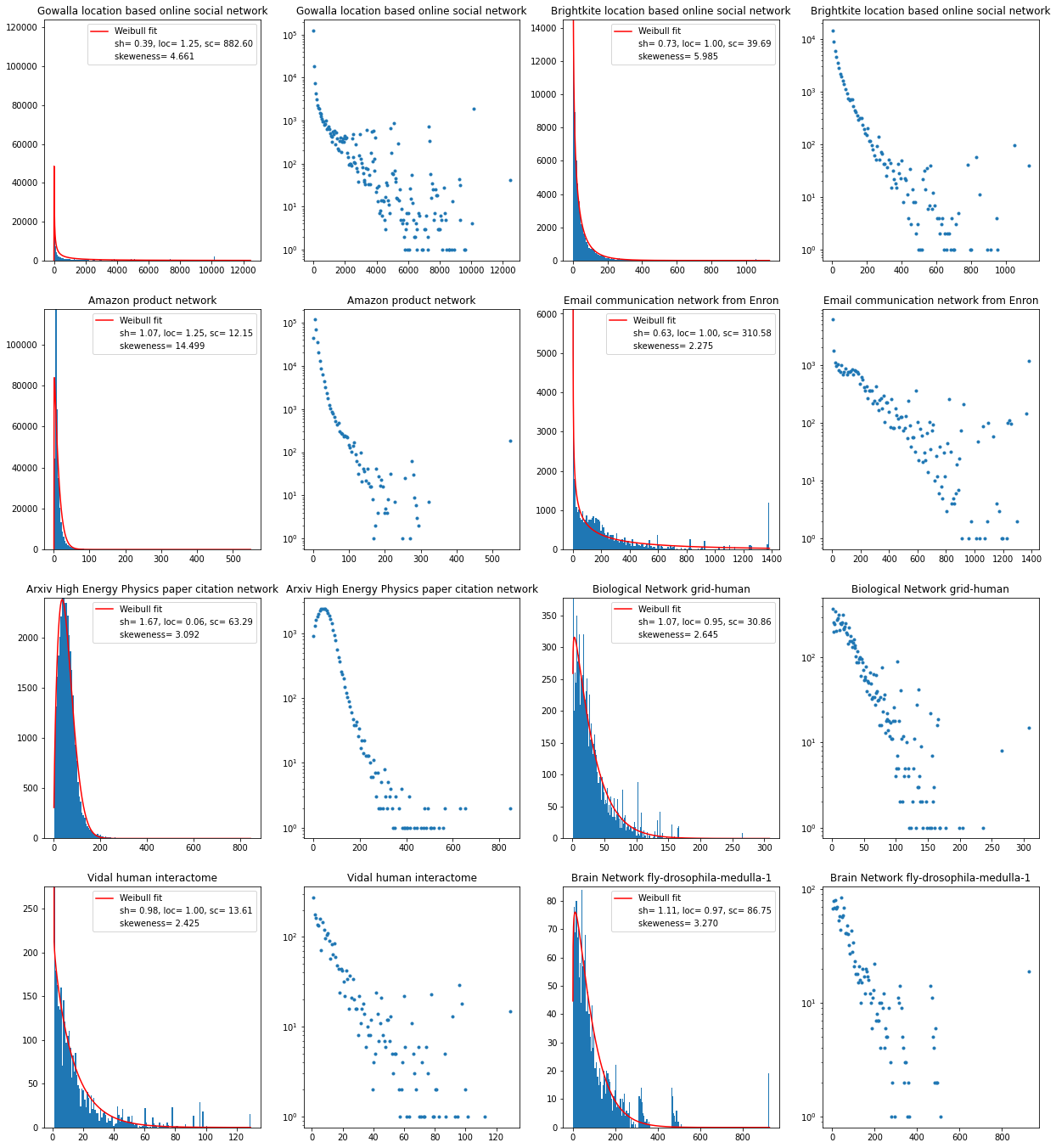}
	\caption{Average neighbor degree distributions (left) and corresponding log-scale for y axes (right) for networks from table~\ref{tab:net}. The x-axis shows the average neighbor degree values, and the y-axis shows the number of vertices.}
\end{figure}

\begin{figure}[h!]
    \centering
	\includegraphics[height = 0.8\textheight, width = 1\textwidth]{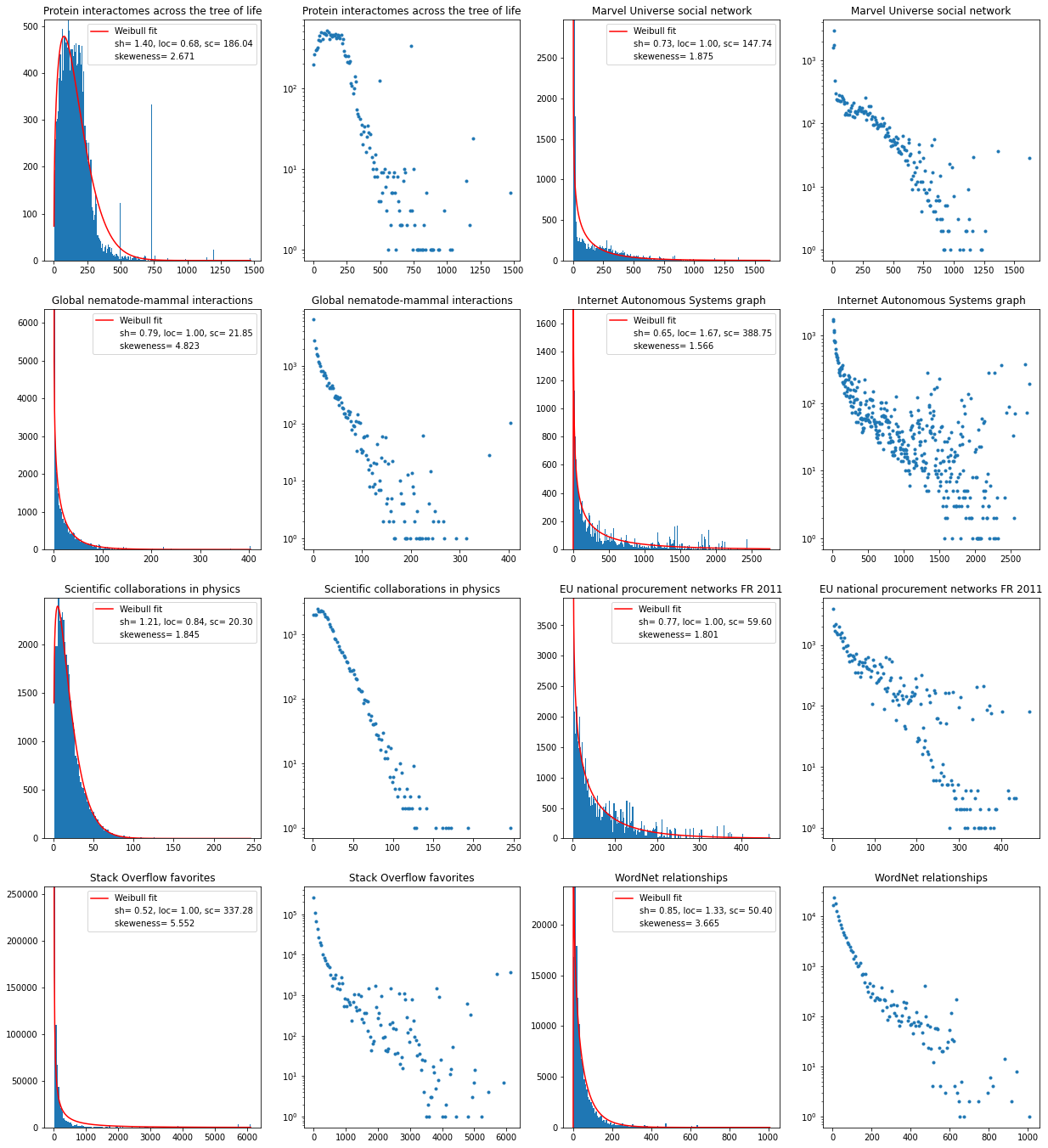}
	\caption{Average neighbor degree distributions (left) and corresponding log-scale for y axes (right) for networks from table~\ref{tab:net}. The x-axis shows the average neighbor degree values, and the y-axis shows the number of vertices.}
    \label{fig:ad5}
\end{figure}

\begin{figure}[h!]
    \centering
	\includegraphics[height = 0.8\textheight, width = 1\textwidth]{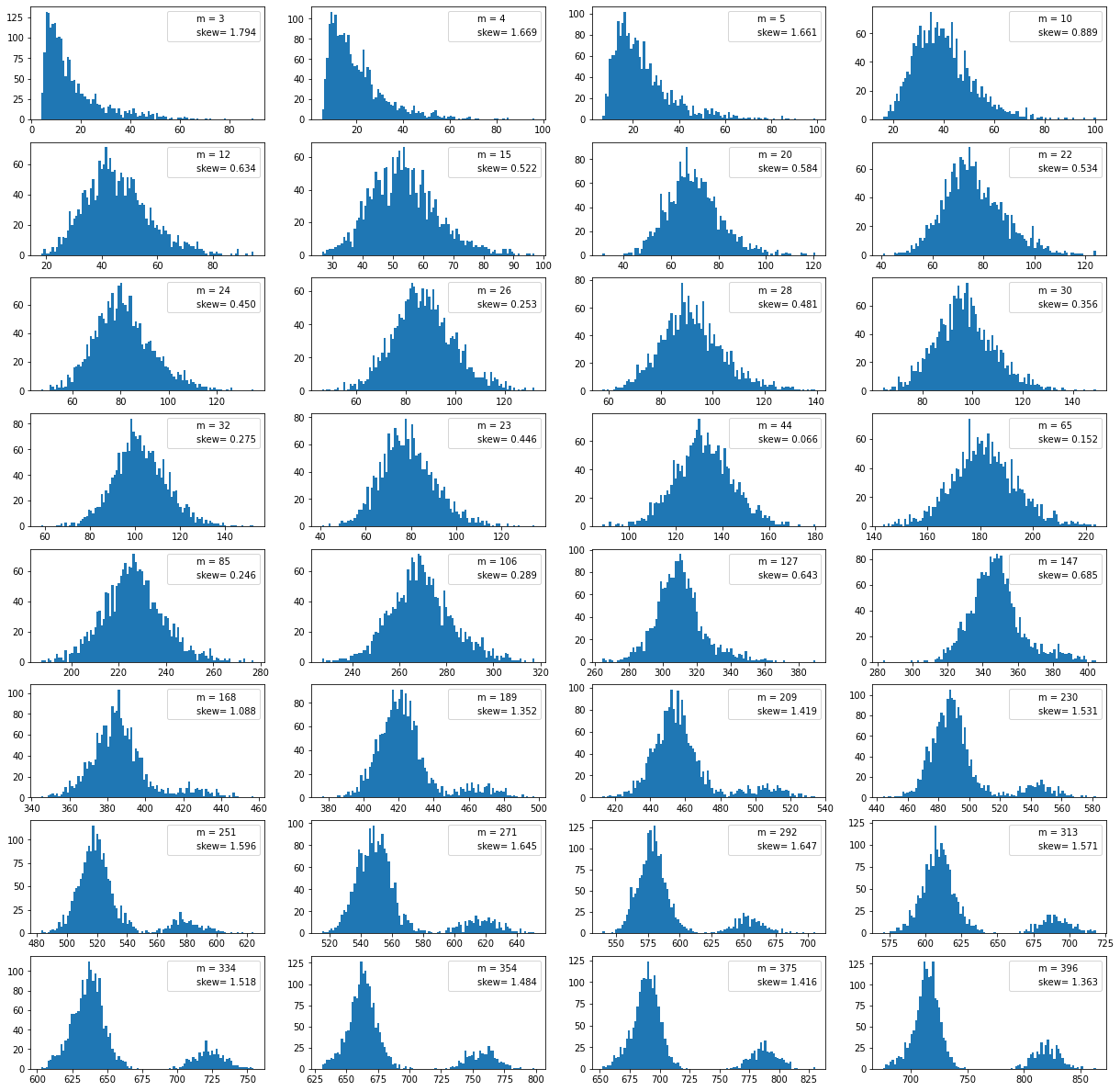}
	\caption{Dependence of the average neighbor degree  distribution and its Pearson’s moment coefficients of skewness on the parameter $m$ for the Barabasi-Albert model $n=2000$. The x-axis shows the average neighbor degree values, and the y-axis shows the number of vertices for distributions.}
    \label{fig:adex}
\end{figure}

\end{document}